\documentclass[12pt]{article}
\usepackage[vmargin=1.18in,hmargin=1in]{geometry}
\usepackage{setspace}
\usepackage{amsmath}
\allowdisplaybreaks  
\usepackage{amssymb}
\usepackage{mathtools}

\usepackage{algorithm}
\usepackage{algpseudocode}

\usepackage{textcomp, gensymb}
\usepackage{gensymb}
\usepackage{enumitem}
\usepackage{graphicx}
\usepackage{scrextend}
\usepackage{blindtext}
\usepackage{caption}
\usepackage{url}
\usepackage{natbib}
\usepackage{subcaption}
\usepackage{circuitikz}
\usepackage{listings}
\usepackage{color}
 \usepackage{microtype}
\usepackage{graphicx}
\usepackage{multirow}
\usepackage{booktabs} 

\usepackage{amssymb}
\usepackage{pifont}
\usepackage{hyperref}
\numberwithin{equation}{section}
\numberwithin{figure}{section}
\numberwithin{table}{section}
\makeatletter
\let\c@figure\c@table
\makeatother

\newcommand{\specialcell}[2][c]{%
\begin{tabular}[#1]{@{}c@{}}#2\end{tabular}}
\usepackage{amsmath}

\usepackage{authblk}
\usepackage{amsmath}
\usepackage{amssymb}
\usepackage{mathtools}

\usepackage{amsthm}
\usepackage[capitalize,noabbrev]{cleveref}
\usepackage{comment}
\usepackage{dsfont}
\theoremstyle{plain}
\newtheorem{theorem}{Theorem}[section]

\newtheorem{lemma}[theorem]{Lemma}

\theoremstyle{definition}
\newtheorem{definition}[theorem]{Definition}
\newtheorem{assumption}[theorem]{Assumption}
\theoremstyle{remark}

\definecolor{codegreen}{rgb}{0,0.6,0}
\definecolor{codegray}{rgb}{0.5,0.5,0.5}
\definecolor{codepurple}{rgb}{0.58,0,0.82}
\definecolor{backcolour}{rgb}{0.95,0.95,0.92}
 
\lstdefinestyle{mystyle}{
    backgroundcolor=\color{backcolour},   
    commentstyle=\color{codegreen},
    keywordstyle=\color{magenta},
    numberstyle=\tiny\color{codegray},
    stringstyle=\color{codepurple},
    basicstyle=\footnotesize,
    breakatwhitespace=false,         
    breaklines=true,                 
    captionpos=b,                    
    keepspaces=true,                 
    numbers=left,                    
    numbersep=5pt,                  
    showspaces=false,                
    showstringspaces=false,
    showtabs=false,                  
    tabsize=2
}
 
\DeclarePairedDelimiter\floor{\lfloor}{\rfloor}

\date{\today}
\usepackage{bbm}

\renewcommand{\P}{\mathbb{P}}

\usepackage{color}

\usepackage{cleveref}
\crefname{assumption}{assumption}{assumptions}
\Crefname{assumption}{Assumption}{Assumptions}

\crefname{theorem}{theorem}{theorems}
\Crefname{theorem}{Theorem}{Theorems}

\crefname{figure}{fig.}{}
\Crefname{figure}{Fig.}{}

\title{\textbf{Density estimation with atoms, and functional estimation for mixed discrete-continuous data}}

\author[1]{Aytijhya Saha}
\author[2]{Aaditya Ramdas}

\affil[1]{Massachusetts Institute of Technology.  \texttt{aytijhya@mit.edu}}
\affil[2]{Carnegie Mellon University. 
\texttt{aramdas@cmu.edu}}
\begin{document}

\date{\today}
\maketitle

\begin{abstract}
In classical density (or density-functional) estimation, it is standard to assume that the underlying distribution has a density with respect to the Lebesgue measure.
However, when the data distribution is a mixture of continuous and discrete components, the resulting methods are inconsistent in theory and perform poorly in practice. In this paper, we point out that a minor modification of existing methods for nonparametric density (functional) estimation can allow us to fully remove this assumption while retaining nearly identical theoretical guarantees and improved empirical performance.
Our approach is very simple: data points that appear exactly once are likely to originate from the continuous component, whereas repeated observations are indicative of the discrete part. Leveraging this observation, we modify existing estimators for a broad class of functionals of the continuous component of the mixture; this modification is a ``wrapper'' in the sense that the user can use any underlying method of their choice for continuous density functional estimation. Our modifications deliver consistency without requiring knowledge of the discrete support, the mixing proportion, and without imposing additional assumptions beyond those needed in the absence of the discrete part. Thus, various theorems and existing software packages can be made automatically more robust, with absolutely no additional price when the data is not truly mixed.

\end{abstract}

\section{Introduction}

Estimating a probability density function or a functional thereof is a fundamental problem in statistics and machine learning. Classical nonparametric approaches such as $k$-nearest neighbor methods, histogram-based estimators, and kernel density estimation \cite{silverman1986density,devroye1985nonparametric} typically assume that the underlying distribution is either absolutely continuous with respect to the Lebesgue measure or purely discrete with respect to the counting measure. However, for many problems, we argue that this may be an entirely avoidable assumption, and one can easily deal with mixed discrete-continuous distributions with a countable number of atoms. Thus, we term the method ``density estimation with atoms''. 

  There is rich literature on estimating functionals of the underlying distribution, such as entropy, mutual information, and divergence measures, but again, these methods either assume fully continuous data~\cite{birge1995estimation,laurent1996efficient,bickel1988estimating,NIPS2015_06138bc5,singh2016finite,moon2017ensemble,moon2018ensemble} or fully discrete data~\cite{antos2001convergence,jiao2017maximum,jiao2015minimax}. When the data come from a mixed discrete-continuous distribution, i.e. a mixture distribution containing both a continuous and discrete component, the presented estimators are inconsistent in theory and perform poorly in practice.


 In this work, we propose a simple approach to this problem. Specifically, observations that appear only once are unlikely to have come from the discrete component (at least at large sample sizes), while repeated observations are almost surely drawn from the discrete part. Leveraging this observation, we isolate the continuous component directly from the data without any prior knowledge of the support or structure of the discrete distribution. Our method is fully nonparametric and adapts to the underlying mixture automatically.

To formalize this idea, suppose we have observations $$X_1,\cdots,X_n\sim (1-\pi_1)F+\pi_1 H_1,$$ where $\pi_1\in(0,1)$, $F$ has density $f$ with respect to the Lebesgue measure and $H_1$ is a discrete distribution with countable support.
We first define the kernel density estimator (KDE) for $f$ based on the unique observations:
\begin{equation}
\label{eq:kde}
     \hat{f}_{\mathcal{U}_n^1}(x) = \frac{1}{(|\mathcal{U}_n^1| \vee 1)h} \sum_{i\in[n]:X_i \in \mathcal{U}_n^1} K\left( \frac{x - X_i}{h} \right),
\end{equation}
 where \( K(\cdot) \) is a kernel function and \( h > 0 \) is the bandwidth parameter and $[n]$ denotes the set of integers from $1$ to $n$ and $$ \mathcal{U}_n^1 = \{ X_i \mid X_i \text{ appears exactly once in } \{X_1, \dots, X_n\}, i\in\{1,\cdots,n\} \}.$$
  Notably, the discrete structure can also be directly estimated from the mixture by assigning point masses at locations with repeated observations, weighted by their empirical frequencies.
  
Standard kernel density estimators (KDEs), including those implemented in widely used softwares (such as \textsf{R}, \textsf{Python}), are designed under the assumption of fully continuous data and fail dramatically in such mixed settings. This failure is vividly demonstrated in \Cref{fig:1}, when applied to samples from a simple mixture of a Gaussian and a Binomial distribution, the naïve KDE fails. In contrast, our simple modification, as discussed above, results in accurate recovery of the underlying density and probability mass function. This motivates our work: to formalize and generalize such atom-aware estimators for a broad class of statistical functionals. Since many real datasets are inherently mixed in nature, 
consistent and efficient estimators that are robust to such heterogeneity could substantially enhance
the reliability of modern data-driven applications. 

 \begin{figure*}[!htb]
\centering
\centering
\subfloat[$n=100$]{\includegraphics[width=0.33\linewidth,height=0.22\linewidth]{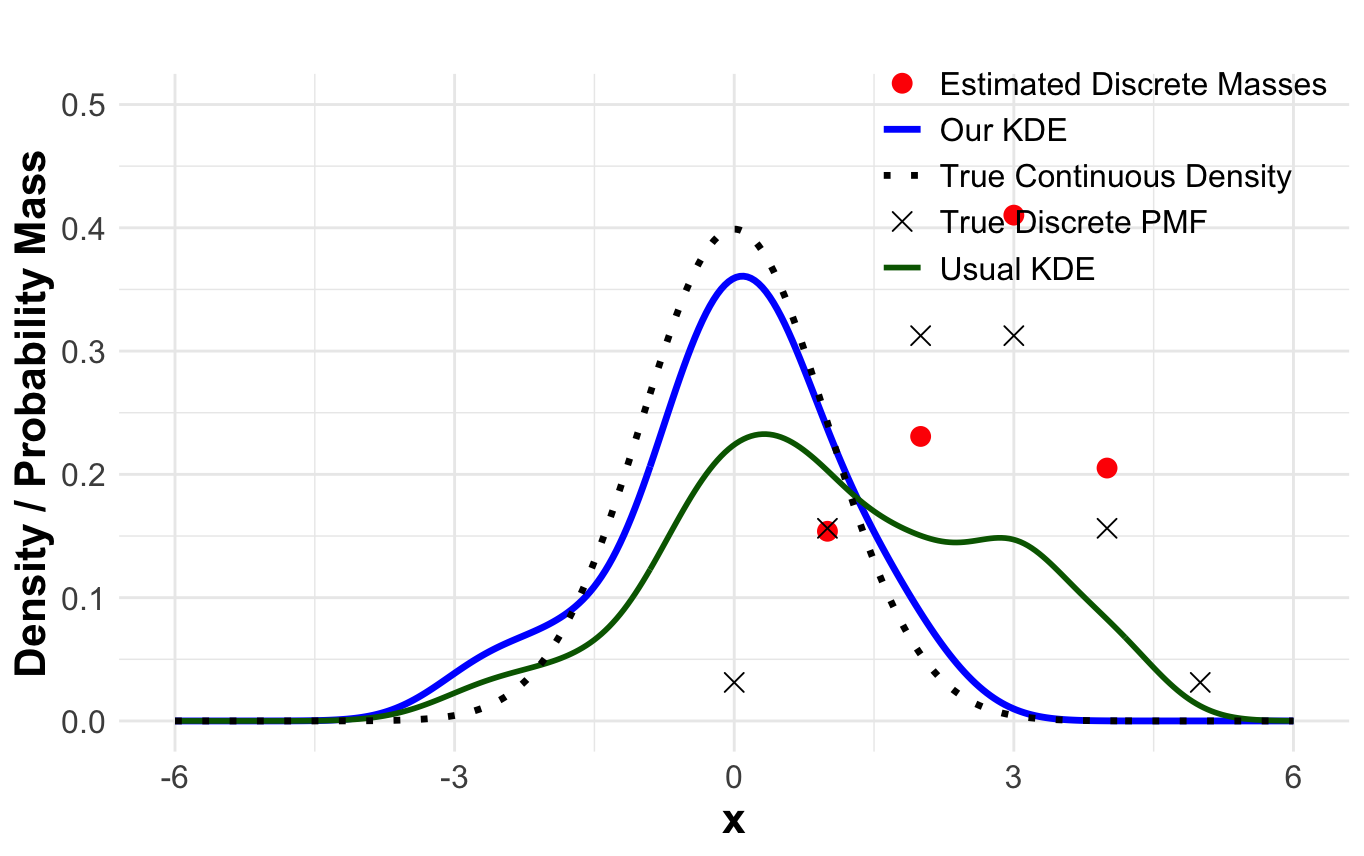}} 
\subfloat[$n=1000$]{\includegraphics[width=0.33\linewidth,height=0.22\linewidth]{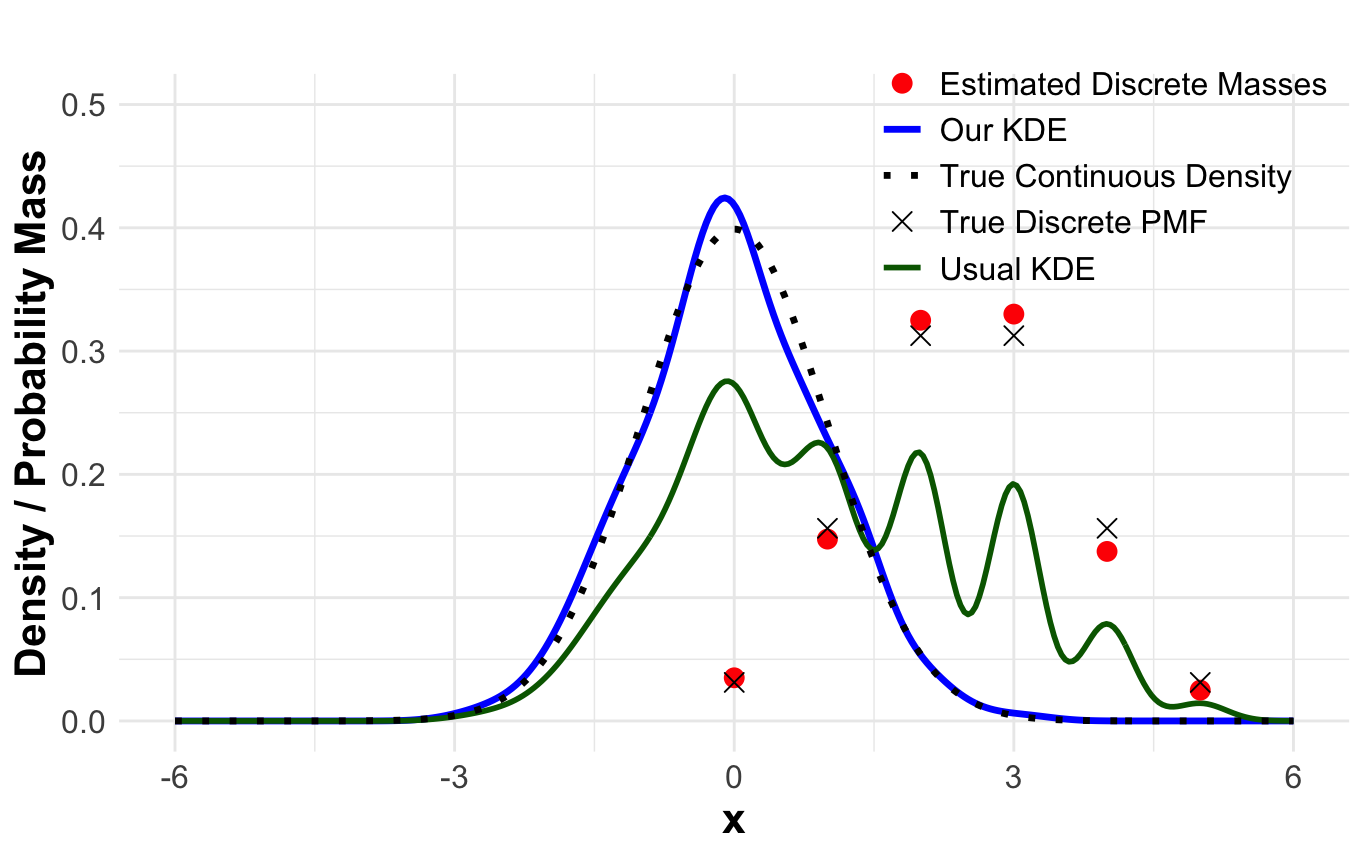}}
\subfloat[$n=10000$]{\includegraphics[width=0.33\linewidth,height=0.22\linewidth]{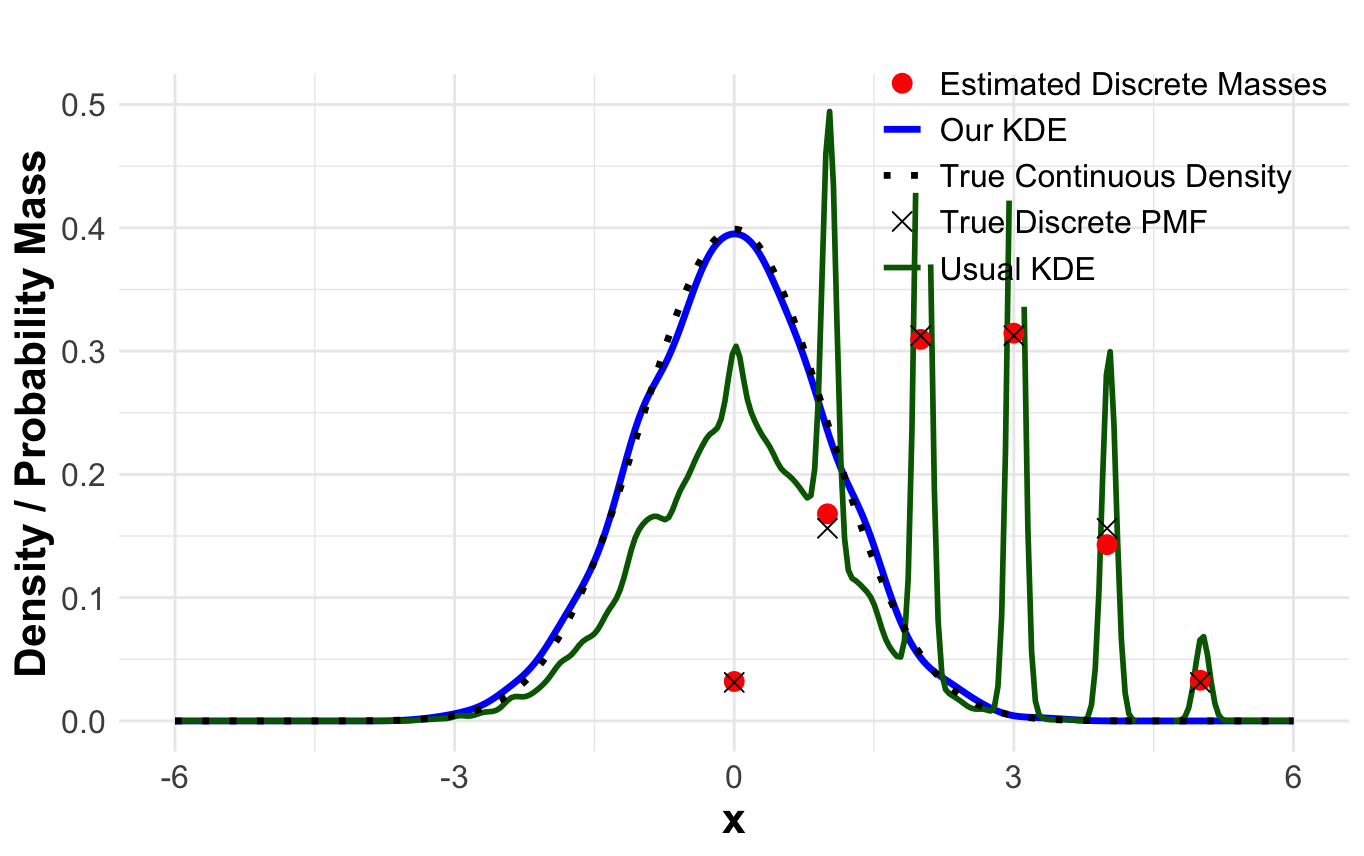}}
\caption[]{Usual KDE (implemented using the \texttt{kde} function from the \texttt{ks} package in \textsf{R}) fails in the presence of atoms. However, our simple modification allows consistent estimation of the density.} 
\label{fig:1}  
\end{figure*}


 Our framework also naturally extends to the estimation of functionals using modern techniques, such as the leave-one-out estimators developed in \cite{NIPS2015_06138bc5}, as we demonstrate in \Cref{sec:meth}.
While we focus on specific estimators to establish theoretical guarantees and illustrate practical performance, our core methodology is not tied to them. It is important to emphasize that the core insight of our approach ---distinguishing between the continuous and discrete components of a distribution based on whether an observation appears uniquely or repeatedly in the sample --- is general and can be readily integrated into a broad range of estimation procedures and statistical problems involving mixed discrete-continuous data.

\subsection{Related works}

\subsubsection{Zero-inflated models}

Zero-inflated models have been extensively studied as a particular instance of mixed discrete-continuous distributions, where the discrete component is a point mass at zero. Such models are common in ecological and biomedical applications, where the variable of interest (e.g., the abundance of a species or the intensity of a clinical measurement) is continuous but exhibits an excess of zero values. Notable contributions include \cite{ancelet2010modelling, lecomte2013compound, liu2019statistical}. However, while these models provide useful insights, they are limited in scope: they typically assume the discrete part consists only of zeros and do not generalize to arbitrary discrete supports as considered in our work.

\subsubsection{Estimation with data mixed discrete-continuous observations}

Our setting is most closely related to works on modeling and estimation from discrete-continuous mixture distributions. \cite{orlitsky2004modeling} and \cite{dragi2017estimating} develop methods for estimating the probability mass function of the discrete component in such mixtures. However, their focus remains confined to characterizing the discrete part, leaving out the estimation of the continuous component or its functionals. Moreover, \cite{marx2021estimating, rahimzamani2018estimators, mesner2020conditional} propose estimators for mutual information and conditional mutual information that can accommodate mixed data. There are two key differences between our work and the preceding papers. First, they focus on mutual information, while we can handle arbitrary functionals. Second, these works propose new estimators to handle the mixed data, while we propose a simple wrapper around any existing estimator that works for continuous data.

\subsubsection{Estimation with data having mixed discrete-continuous features}
There is a substantial body of work on statistical estimation and learning in settings where the feature space is comprised of both discrete and continuous variables, see e.g., \cite{li2021nonparametric} and \cite{bhadra2018inferring}.
However, this setting is fundamentally different from ours: while they address mixed-type features (i.e., columns), our work deals with mixed-type observations (i.e., rows) --- where the observed data itself is drawn from a hybrid distribution over a union of discrete and continuous domains. 

 \begin{figure*}[!htb]
\centering
\centering
\subfloat[Mixed discrete-continuous observations]{\includegraphics[width=0.45\linewidth]{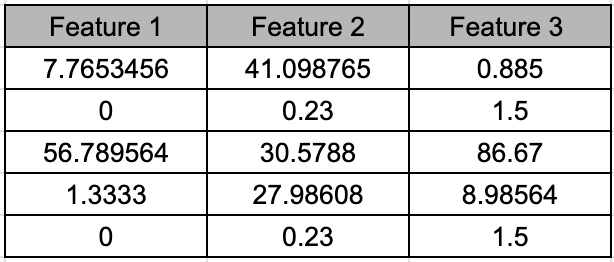}} \hspace{1cm}
\subfloat[Mixed discrete-continuous features]{\includegraphics[width=0.45\linewidth]{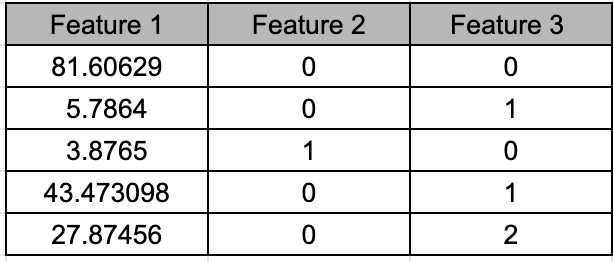}}
\caption[]{In this paper, we focus on the setup on the left side only. Each table has 5 data points from a three-dimensional distribution. The left table shows data with mixed discrete-continuous \emph{observations}: each datapoint comes either from a distribution with a density, or from a discrete distribution with unknown support. The right table shows data with mixed discrete-continuous \emph{features}: each feature is either discrete (categorical) or continuous (real-valued).} 
\label{fig:2}  
\end{figure*}

\subsubsection{Huber-robust estimation.} Robust estimation under contamination has a rich history, with the Huber contamination model \citep{huber1964robust,huber1965robust}, where a fraction of the data is assumed to be corrupted by an arbitrary distribution. Several works, including \cite{liu2017density,uppal2020robust}, have proposed density estimation methods under the Huber-contamination model. Although superficially similar to our setting, where a portion of the data arises from a discrete component (which can be viewed as the contamination part in the Huber-contamination model), there are key differences. First, existing Huber-robust estimators are typically inconsistent for the uncontaminated target distribution. In contrast, our proposed method achieves consistency by adaptively identifying and separating the discrete and continuous components. Secondly, Huber-robust procedures often assume the contamination level (i.e., proportion of corrupted samples) is known, whereas our approach adaptively estimates this proportion from the data. Finally, Huber-robust methods pay a price (and are not optimal) when the data has no contamination, as they are designed to guard against worst-case scenarios. However, our method incurs no additional cost when the data is purely continuous, thereby retaining optimality in the absence of discrete contamination.

 \subsection{Our contributions and paper outline.} 
 Our main contributions are as follows:
\begin{itemize}
    \item We propose a general framework for nonparametric estimation of the density and density functionals corresponding to the continuous component using data generated from a discrete-continuous mixture.
    \item We show that the modified KDE \eqref{eq:kde} is consistent and achieves minimax optimal mean integrated absolute error (MIAE), under the standard assumptions of the purely continuous setting, in the presence of atoms in the data distribution. 
   \item We provide rigorous theoretical guarantees for the consistency of these estimators for the density functionals without making additional assumptions. Our estimators still achieve $n^{-1/2}$ consistency whenever the density of the continuous component is sufficiently smooth and the support of the discrete component is allowed to grow in a triangular-array set-up. We have consistency even when the support of the discrete component is countable.
    \item We demonstrate empirically that our approach performs well in practice, while the standard methods fail in the presence of atoms.
 \end{itemize}

  The remainder of the paper is organized as follows. In \Cref{sec:kde}, we discuss the consistency and convergence rate of our modified KDE (defined in \eqref{eq:kde}) in the presence of atoms. In \Cref{sec:meth}, we introduce our framework and methodology for estimating the density functionals for discrete-continuous mixtures. \Cref{sec:theory} presents the theoretical analysis of our estimators for the functionals, establishing consistency and convergence rates. We report empirical results demonstrating the effectiveness of our approach in \Cref{sec:expt}. We discuss the future directions and conclude the article in \Cref{sec:conc}. Detailed proofs of the theoretical results and some additional experimental results are provided in the Appendix.

\section{Consistent density estimation in the presence of atoms}
\label{sec:kde}
Some smoothness assumptions on the densities are required to study the convergence properties of the KDE. Here we assume the Hölder smoothness, which is a standard in nonparametric literature.
 
 \begin{definition}
    Let \( X \subset \mathbb{R}^d \) be a compact space. For any multi-index \( r = (r_1, \ldots, r_d) \), with \( r_i \in \mathbb{N} \), define \( |r| = \sum_i r_i \), and let 
$D^r = \frac{\partial^{|r|}}{\partial x_1^{r_1} \cdots \partial x_d^{r_d}}.$
The Hölder class \( \Sigma(s, L) \) is the set of functions \( f \in L^2(X) \) satisfying
$|D^r f(x) - D^r f(y)| \leq L \|x - y\|^{s - |r|}$
for all multi-indices \( r \) such that \( |r| \leq \lfloor s \rfloor \), and for all \( x, y \in X \).
Moreover, define the Bounded Hölder class \( \Sigma(s, L, B_0, B) \) to be
$\left\{ f \in \Sigma(s, L) : B_0 < f < B \right\}$.
\end{definition}
This smoothness assumption allows us to quantify the convergence behavior of the KDE in terms of the mean integrated absolute error (MIAE), which is a widely used performance metric in density estimation \cite{devroye1985nonparametric,HALL198859}.
The next theorem shows that our modified KDE \eqref{eq:kde} is consistent and achieves the minimax optimal rate $\mathcal{O}(n^{-\frac{s}{2s+d}})$ for MIAE for fully continuous settings \cite{devroye1985nonparametric}, when the discrete part has finite support, which is allowed to grow with the sample size, in a triangular-array set-up.
\begin{theorem}
 \label{thm:kde}
Suppose that $X_1,\cdots,X_n\sim (1-\pi_1)F+\pi_1 H_1,$ where $\pi_1\in(0,1)$, $F$ has density $f$ with respect to the Lebesgue measure and $H_1$ is any discrete distribution with countable support. If $nh\to\infty$ and $h\to0$ as $n\to\infty$, then, the estimator  \( \hat{f}_{\mathcal{U}_n^1}(x) \) defined in \eqref{eq:kde} satisfies
$$\mathbb E\left( \int|\hat{f}_{\mathcal{U}_n^1}(x)-{f}(x)|dx\right) \to 0\quad \text{as } n \to \infty.
 $$ 
 Further, suppose $H_1$ has finite support $\mathcal{S}_n$, which may grow with $n$, and let its probability mass function (p.m.f.) be $\{p^{(n)}_s\}_{s\in \mathcal{S}_n}$. Assume that the minimum mass of an atom satisfies $\min_{s\in \mathcal{S}_n}p_s^{(n)}\geq\frac{1}{\pi_1}(1-(cn^{-\frac{2s}{2s+d}})^\frac{1}{n-1})$, for some constant $c>0$. Let $K$ be a kernel of order $\floor{s}$ satisfying $\int K^2(u)du <\infty$ and $\int |u|^\beta|K(u)|du <\infty$, $f\in \Sigma(s, L)$, $h=\alpha n^{-\frac{1}{2s+d}}$, for some $\alpha>0$. Then,
  \[
 \mathbb E\left( \int|\hat{f}_{\mathcal{U}_n^1}(x)-{f}(x)|dx\right) =\mathcal{O}(n^{-\frac{s}{2s+d}}).
  \]
 \end{theorem}
Note that the above assumption on the p.m.f. is trivially satisfied if $H_1$ has a fixed support $S$ that does not change with $n$ (i.e., we are not in a triangular array setup). 
So, our simple strategy—focusing on unique observations --- does not sacrifice statistical efficiency in the continuous regime.

In contrast to standard KDE, which can be severely biased near atoms (as illustrated in \Cref{fig:1}), our estimator effectively disentangles the discrete and continuous parts. It works automatically, without requiring prior knowledge of the atom locations or proportions, and adapts to the structure of the data.

\section{Estimation of density functionals}
\label{sec:meth}
Having discussed the core intuition behind our approach, we now focus on estimating density functionals of the continuous component in a mixed discrete-continuous distribution. Examples of such functionals include entropy, mutual information, and divergence measures, which are widely used in statistics, information theory, and machine learning. We first review some estimators designed for fully continuous data, such as those proposed in \cite{NIPS2015_06138bc5}, and then show how they can be extended and adapted to the mixed data setting. These modifications are simple yet powerful: they preserve the statistical guarantees of the original estimators while making them robust to the presence of atoms.
\subsection{Preliminaries}
Let
$F$ and $G$ be measures over a compact space $\mathcal{X}\subseteq\mathbb R^d$ that are absolutely continuous w.r.t the Lebesgue measure. Let $f,g \in L_2(\mathcal{X})$ be the
density (Radon-Nikodym derivatives) with respect to the Lebesgue measure. Given observations 
\begin{equation}
\label{prev-setup}
    X_1,\cdots,X_n\stackrel{i.i.d.}{\sim}F ~~\text{ and }~~ Y_1,\cdots,Y_m\stackrel{i.i.d.}{\sim}G,
\end{equation}
\cite{NIPS2015_06138bc5} develops a recipe for estimating statistical
functionals of one or more nonparametric distributions of the form
\begin{equation}
\label{eq:func_form}
    T(F)=T(f)=\phi\left(\int\nu(f)d\mu\right) ~~~\text{ or }~~~ T(F,G)=T(f,g)=\phi\left(\int\nu(f,g)d\mu\right),
\end{equation}
where $\phi$ and $\nu$ are real-valued Lipschitz functions that are twice differentiable. They use the following functional Taylor expansion on the densities
\begin{equation}
\label{eq:tylor}
    T(f)=T(g)+\mathbb E_F\psi(X;g)+ \mathcal{O}(\|f-g\|^2),
\end{equation}
where $\psi$ is the influence function, which is defined in terms of the Gâteaux derivative by $$\psi(x;F)=\frac{\partial T((1-t)F+t\delta_x)}{\partial t}\big|_{t=0},$$ where $\delta_x$ is the dirac delta function at $x$. They study data-splitting (DS) and leave-one-out (LOO) type estimators and analyze their
convergence. For the DS estimator, half of the data is used to compute the density estimator, and the remaining half is used to compute the sample mean of the influence function.
\begin{equation}
    \label{eq-ds-original}
\hat{T}^{(1)}_{\text{DS}}=T(\hat{f}^{(1)})+\frac{1}{n/2}\sum_{i=\floor{n/2}+1}^n \psi(X_i;\hat{f}^{(1)})
\end{equation}
and $\hat{T}^{(2)}_{\text{DS}}$ is defined similarly. The final estimator is $\hat{T}_{\text{DS}} =   (\hat{T}^{(1)}_{\text{DS}}+\hat{T}^{(2)}_{\text{DS}})/2$.
They propose a Leave-One-Out (LOO) version of the above
estimator
\begin{equation}
    \label{eq-loo-original}
    \hat{T}_{\text{LOO}}=\frac{1}{n}\sum_{i=1}^n (T(\hat{f}_{-i})+\psi(X_i;\hat{f}_{-i})),
\end{equation}
where $\hat{f}_{-i}$ is a density estimate using all the samples  except for $X_i$.

Akin to the one distribution case, they propose the following DS and LOO versions for the two distribution case.
\begin{equation}
    \label{eq-ds-original-2}
\hat{T}^{(1)}_{\text{DS}}=T(\hat{f}^{(1)},\hat{g}^{(1)})+\frac{1}{n/2}\sum_{i=\floor{n/2}+1}^n \psi_f(X_i;\hat{f}^{(1)},,\hat{g}^{(1)})+\frac{1}{m/2}\sum_{i=\floor{m/2}+1}^m\psi_g(Y_i;\hat{f}^{(1)},,\hat{g}^{(1)}),
\end{equation}
\begin{equation}
    \label{eq-loo-original-2}
    \hat{T}_{\text{LOO}}=\frac{1}{\max(n,m)}\sum_{i=1}^{\max(n,m)} (T(\hat{f}_{-i},\hat{g}_{-i})+\psi_f(X_i;\hat{f}_{-i},\hat{g}_{-i})+\psi_g(Y_i;\hat{f}_{-i},\hat{g}_{-i})).
\end{equation}
For the LOO estimator, if $n > m$, the points $Y_1,\cdots,Y_m$ are cycled through until all $X_i$'s have been summed over, or vice versa.
These estimators are not consistent in the presence of atoms in the distribution. 
In the following subsection, we propose a simple modification of the above estimators that can consistently estimate density functionals of the continuous part in the presence of (countably many) atoms in the data distributions. 

\subsection{Our extension}

In contrast to the classical set-up where data arises from either a continuous or a discrete distribution, we have samples from a discrete-continuous mixture
\begin{equation}
    X_1,\cdots,X_n\stackrel{i.i.d.}{\sim} (1-\pi_1)F+\pi_1 H_1  ~~\text{ and }~~ Y_1,\cdots,Y_m\stackrel{i.i.d.}{\sim} (1-\pi_2)G+\pi_2 H_2,
\end{equation}
where $\pi_1,\pi_2\in(0,1)$ are unknown constants, $F,G$ have Lebesgue densities $f,g$ respectively and $H_1,H_2$ are discrete distributions having countable supports. We develop estimators for functionals \eqref{eq:func_form} using data generated from the above mixed distributions.


Define the sets of unique observations as 
\begin{equation}
    \mathcal{U}_n^1 = \{ X_i \mid X_i \text{ appears exactly once in } \{X_1, \dots, X_n\}, i\in\{1,\cdots,n\} \},
\end{equation}
\begin{equation}
    \mathcal{U}_m^2 = \{ Y_i \mid Y_i \text{ appears exactly once in } \{Y_1, \dots, Y_m\}, i\in\{1,\cdots,m\} \}.
\end{equation}

Analogous to \eqref{eq-ds-original} and \eqref{eq-ds-original-2}, we split $\mathcal{U}_n^1$  into two parts: $\mathcal{U}_n^{1,1}:=\{X_i: i\leq \floor{n/2}, X_i\in\mathcal{U}_n^1\}$ and $\mathcal{U}_n^{1,2}:=\{X_i: i\geq \floor{n/2}+1, X_i\in\mathcal{U}_n^1\}$ and similarly split $\mathcal{U}_m^2$ into $\mathcal{U}_m^{2,1}:=\{X_i: i\leq \floor{m/2}, X_i\in\mathcal{U}_m^2\}$ and $\mathcal{U}_m^{2,2}:=\{X_i: i\geq \floor{m/2}+1, X_i\in\mathcal{U}_m^2\}$. And our DS estimators are defined below; the first part is used to compute the density estimator, and the remaining part is used to compute the sample mean of the influence function:
\begin{equation}
\label{eq-ds}
\hat{T}^{\text{DS},1}_{\mathcal{U}_n^1}=T(\hat{f}_{\mathcal{U}_n^{1,1}})+\frac{1}{|\mathcal{U}_n^{1,2}| \vee 1} \sum_{X_i \in \mathcal{U}_n^{1,2}}\psi(X_i;\hat{f}_{\mathcal{U}_n^{1,1}}),
\end{equation}
\begin{equation}
\label{eq-ds-2}
\hat{T}^{\text{DS},1}_{\mathcal{U}_n^1,\mathcal{U}_m^2}=T(\hat{f}_{\mathcal{U}_n^{1,1}},\hat{g}_{\mathcal{U}_m^{2,1}})+\frac{\sum_{X_i \in \mathcal{U}_n^{1,2}}\psi_f(X_i;\hat{f}_{\mathcal{U}_n^{1,1}},\hat{g}_{\mathcal{U}_m^{2,1}})}{|\mathcal{U}_n^{1,2}| \vee 1}+\frac{\sum_{Y_i \in \mathcal{U}_m^{2,2}}\psi_g(Y_i;\hat{f}_{\mathcal{U}_n^{1,1}},\hat{g}_{\mathcal{U}_m^{2,1}})}{|\mathcal{U}_m^{2,2}| \vee 1}.
\end{equation}
 Similarly, we have $\hat{T}^{\text{DS},2}_{\mathcal{U}_n^1}$ and $\hat{T}^{\text{DS},2}_{\mathcal{U}_n^1,\mathcal{U}_m^2}$. Our final DS estimators are defined as $$\hat{T}^{\text{DS}}_{\mathcal{U}_n^1}=\frac{\hat{T}^{\text{DS},1}_{\mathcal{U}_n^1}+\hat{T}^{\text{DS},2}_{\mathcal{U}_n^1}}{2}~~\text{ and }~~\hat{T}^{\text{DS}}_{\mathcal{U}_n^1,\mathcal{U}_m^2}=\frac{\hat{T}^{\text{DS},1}_{\mathcal{U}_n^1,\mathcal{U}_m^2}+\hat{T}^{\text{DS},2}_{\mathcal{U}_n^1,\mathcal{U}_m^2}}{2}.$$
 
 Now, we define the following LOO estimators, which are
analogous to \eqref{eq-loo-original} and \eqref{eq-loo-original-2}:
\begin{equation}
\label{eq-loo}
    \hat{T}^{\text{LOO}}_{\mathcal{U}_n^1}=\frac{1}{|\mathcal{U}_n^1| \vee 1} \sum_{i:X_i \in \mathcal{U}_n^1}\left(T(\hat{f}_{\mathcal{U}_n^1}^{(-i)})+\psi(X_i;\hat{f}_{\mathcal{U}_n^1}^{(-i)})\right),
\end{equation}
\begin{equation}
\label{eq-loo-2}
\hat{T}^{\text{LOO}}_{\mathcal{U}_n^1,\mathcal{U}_m^2}=\frac{1}{|\mathcal{U}_n^1|\vee|\mathcal{U}_m^2| \vee 1} \sum_{i=1}^{|\mathcal{U}_n^1|\vee|\mathcal{U}_m^2|}\left(T(\hat{f}_{\mathcal{U}_n^1}^{(-j_i)},\hat{g}_{\mathcal{U}_m^2}^{(-k_i)})+\psi_f(X_i;\hat{f}_{\mathcal{U}_n^1}^{(-j_i)},\hat{g}_{\mathcal{U}_m^2}^{(-k_i)})+\psi_g(Y_i;\hat{f}_{\mathcal{U}_n^1}^{(-j_i)},\hat{g}_{\mathcal{U}_m^2}^{(-k_i)})\right),
\end{equation}
where $j_1<j_2<\cdots<j_{|\mathcal{U}_n^1|}$ are indices of the $X_i$s which are in $\mathcal{U}_n^1$ and $k_1<k_2<\cdots<k_{|\mathcal{U}_m^2|}$ are indices of the $Y_i$s which are in $\mathcal{U}_m^2$. Here, for some subset $\mathcal{A}$ of $\{X_1,\cdots,X_n\}$, $\hat{f}_{\mathcal{A}}$ denotes the kernel density estimator using elements in $\mathcal{A},$ i.e., 
\begin{equation}
    \hat{f}_{\mathcal{A}}(x) = \frac{1}{(|\mathcal{A}| \vee 1)h} \sum_{X_i \in \mathcal{A}} K\left( \frac{x - X_i}{h} \right),
\end{equation}
where \( K(\cdot) \) is a kernel function and \( h > 0 \) is the bandwidth parameter, and for some $j\in\{1,\cdots,n\},\hat{f}_{\mathcal{A}}^{(-j)}$ denotes the kernel density estimator using elements in $\mathcal{A}\setminus\{X_j\}.$  Similarly, for some subset $\mathcal{B}$ of $\{Y_1,\cdots,Y_n\}$, $\hat{g}_{\mathcal{B}}$ denotes the kernel density estimator using elements in $\mathcal{B}$ and for some $k\in\{1,\cdots,m\},\hat{g}_{\mathcal{B}}^{(-k)}$ denotes the kernel density estimator using elements in $\mathcal{B}\setminus\{Y_k\}.$ 

Although the above discussion focuses on extending a particular class of estimators using influence functions, we reemphasize that the core idea underlying our approach is broadly applicable and can be extended to a wider range of estimators and problems beyond this specific setting.

\section{Asymptotic properties of density functional estimators}
\label{sec:theory}
In this section, we focus on establishing the asymptotic properties of our estimators of density functionals. Remarkably, when it comes to proving consistency, we do not require any new assumptions beyond those used in \cite{NIPS2015_06138bc5}, even in the presence of countably many atoms. To study convergence rates, however, we consider the triangular array set-up, where for each sample size $n$, the discrete components have finite support that is allowed to grow with $n$. A similar triangular array setup is also required for deriving the convergence rate of our KDE, as previously discussed in \Cref{thm:kde}.

In what follows, we make the following regularity condition on the influence function, which corresponds to Assumption 4 in \cite{NIPS2015_06138bc5} and is essential for establishing the theoretical results.

\begin{assumption}
\label{assump1}    
For a functional $T(f)$ of one distribution, the influence function  $\psi$ satisfies
    \begin{equation}
    \label{eq:assmp-1d}
    \mathbb E\left[\left(\psi(X;{f}^\prime)-\psi(X;{f})\right)^2\right]= \mathcal{O}(\|{f}^\prime-f\|^2),
\end{equation}
and for a functional $T(f,g)$ of two distributions, the influence functions $\psi_f,\psi_g$ satisfy
\begin{align}
\label{eq:assmp-2d-1}
    \mathbb E_f\left[\left(\psi_f(X;{f}^\prime,{g}^\prime)-\psi_f(X;{f},g)\right)^2\right]= \mathcal{O}(\|{f}^\prime-f\|^2+\|{g}^\prime-g\|^2), \text{ as } \|{f}^\prime-f\|,\|{g}^\prime-g\|\to 0.\\
    \label{eq:assmp-2d-2}
    \mathbb E_g\left[\left(\psi_g(X;{f}^\prime,{g}^\prime)-\psi_g(X;{f},g)\right)^2\right]= \mathcal{O}(\|{f}^\prime-f\|^2+\|{g}^\prime-g\|^2), \text{ as } \|{f}^\prime-f\|,\|{g}^\prime-g\|\to 0.
\end{align}
\end{assumption}


We now state the results for the one-sample estimator $ \hat{T}^{\text{DS}}_{\mathcal{U}_n^1}$.
\begin{theorem} 
\label{thm:conv-ds}
Let $f\in \Sigma(s, L, B_0, B)$ and  $\psi$ satisfy 
\Cref{assump1}. Then, $\mathbb E|\hat{T}^{\text{DS}}_{\mathcal{U}_n^1}- T(F)|\to 0$.  
 Further, suppose $H_1$ has finite support $\mathcal{S}_n$, which may grow with $n$, and let its probability mass function (p.m.f.) be $\{p^{(n)}_s\}_{s\in \mathcal{S}_n}$. Assume that the minimum mass of an atom satisfies $\min_{s\in \mathcal{S}_n}p_s^{(n)}\geq\frac{1}{\pi_1}(1-(cn^{-\frac{6s}{2s+d}})^\frac{1}{n-1})$, for some constant $c>0$. Then $\mathbb E|\hat{T}^{\text{DS}}_{\mathcal{U}_n^1}- T(F)|$ is $\mathcal{O}\left(n^{\frac{-2s}{2s+d}}\right)$ if $s < d/2$ and $\mathcal{O}(n^{-1/2})$ when $s \geq d/2$. Additionally, when $H_1$ has fixed finite support $S$, $s > d/2$ and $\psi\neq 0$, for $i=1,2,$
 \begin{equation}
 \label{conv-ds}
\sqrt{n}  \left( \hat{T}^{\text{DS}}_{\mathcal{U}_n^1} - T(F) \right)    \xrightarrow{d}N\left(0,\frac{1}{1-\pi_1}\mathbb V_f(\psi(X,f))\right), \text{ as } n\to\infty.
\end{equation}
\end{theorem}
Notably, the assumptions in the above theorem are identical to those in Theorem 14 of \cite{NIPS2015_06138bc5}. While the original estimator enjoys $L_2$ convergence under purely continuous settings, our analysis guarantees only $L_1$ convergence due to the added complexity introduced by the presence of atoms in the distribution. A similar result holds for the two-sample estimator as well, under analogous assumptions.
\begin{theorem}
 \label{thm:conv-ds-2}
If $f,g\in \Sigma(s, L, B_0, B)$ and  $\psi_f,\psi_g$ satisfy 
\Cref{assump1}, then $\mathbb E|\hat{T}^{\text{DS}}_{\mathcal{U}_n^1,\mathcal{U}_m^2}- T(F,G)|\to 0$. Further, suppose $H_1$ and $H_2$ have finite supports $\mathcal{S}_n$ and $\mathcal{S}_m^\prime$, which may grow with $n$ and $m$, and let their probability mass functions (p.m.f.) be $\{p^{(n)}_s\}_{s\in \mathcal{S}_n}$ and $\{q^{(m)}_s\}_{s\in \mathcal{S}_m^\prime}$ respectively. Assume that the minimum masses of an atom satisfy $\min_{s\in \mathcal{S}_n}p_s^{(n)}\geq\frac{1}{\pi_1}(1-(cn^{-\frac{6s}{2s+d}})^\frac{1}{n-1})$ and $\min_{s\in \mathcal{S}_m^\prime}q_s^{(m)}\geq\frac{1}{\pi_2}(1-(cm^{-\frac{6s}{2s+d}})^\frac{1}{m-1})$, for some constant $c>0$. Then, $\mathbb E|\hat{T}^{\text{DS}}_{\mathcal{U}_n^1,\mathcal{U}_m^2}- T(F,G)|$ is $\mathcal{O}\left(n^{\frac{-2s}{2s+d}}+m^{\frac{-2s}{2s+d}}\right)$ if $s < d/2$ and $\mathcal{O}(n^{-1/2}+m^{-1/2})$ when $s \geq d/2$. Additionally, when $H_1$ has fixed finite support $S$, $s > d/2$ and $\psi_f,\psi_g\neq 0$,
 \begin{equation}
 \label{conv-ds-2}
\sqrt{n}(\hat{T}^{\text{DS}}_{{\mathcal{U}_n^1,\mathcal{U}_m^2}}-T(F,G))\xrightarrow{d}
N\left(0, \frac{1}{\zeta(1-\pi_1)}\mathbb V_f(\psi_f(X;f,g)))+\frac{1}{(1-\zeta)(1-\pi_2)}\mathbb V_g(\psi_g(X;f,g))\right),
 \end{equation}
as $n,m\to\infty$ in such way that $n/(n+m) \to \zeta\in(0, 1)$. 
\end{theorem}
Having established consistency and asymptotic properties of our DS estimators, we now turn our attention to the LOO estimators and state the corresponding results.
\begin{theorem}
\label{thm:conv-loo}
Let $f\in \Sigma(s, L, B_0, B)$ and  $\psi$ satisfy 
\Cref{assump1}. Then,
 $\mathbb E|\hat{T}^{\text{LOO}}_{\mathcal{U}_n^1}- T(F)|\to 0$, as $n\to\infty$. Further, suppose $H_1$ has finite support $\mathcal{S}_n$, which may grow with $n$, and let its probability mass function (p.m.f.) be $\{p^{(n)}_s\}_{s\in \mathcal{S}_n}$. Assume that the minimum mass of an atom satisfies $\min_{s\in \mathcal{S}_n}p_s^{(n)}\geq\frac{1}{\pi_1}(1-(cn^{-\frac{6s}{2s+d}})^\frac{1}{n-1})$, for some constant $c>0$. Then $\mathbb E|\hat{T}^{\text{LOO}}_{\mathcal{U}_n^1}- T(F)|$ is $\mathcal{O}\left(n^{\frac{-2s}{2s+d}}\right)$ if $s < d/2$ and $\mathcal{O}(n^{-1/2})$ when $s \geq d/2$.
\end{theorem}
We now move to LOO estimators for functionals of two distributions.
\begin{theorem}
\label{thm:conv-loo-2}
Let $f,g\in \Sigma(s, L, B_0, B)$ and  $\psi_f,\psi_g$ satisfy 
\Cref{assump1}. Then,
 $\mathbb E|\hat{T}^{\text{LOO}}_{\mathcal{U}_n^1,\mathcal{U}_m^2}- T(F,G)|\to 0$, as $n,m\to\infty$.  Further, suppose $H_1$ and $H_2$ have finite supports $\mathcal{S}_n$ and $\mathcal{S}_m^\prime$, which may grow with $n$ and $m$, and let their probability mass functions (p.m.f.) be $\{p^{(n)}_s\}_{s\in \mathcal{S}_n}$ and $\{q^{(m)}_s\}_{s\in \mathcal{S}_m^\prime}$ respectively. Assume that the minimum masses of an atom satisfy $\min_{s\in \mathcal{S}_n}p_s^{(n)}\geq\frac{1}{\pi_1}(1-(cn^{-\frac{6s}{2s+d}})^\frac{1}{n-1})$ and $\min_{s\in \mathcal{S}_m^\prime}q_s^{(m)}\geq\frac{1}{\pi_2}(1-(cm^{-\frac{6s}{2s+d}})^\frac{1}{m-1})$, for some constant $c>0$. Then, $\mathbb E|\hat{T}^{\text{LOO}}_{\mathcal{U}_n^1,\mathcal{U}_m^2}- T(F,G)|$ is $\mathcal{O}\left(n^{\frac{-2s}{2s+d}}+m^{\frac{-2s}{2s+d}}\right)$ if $s < d/2$ and $\mathcal{O}(n^{-1/2}+m^{-1/2})$ when $s \geq d/2$.
\end{theorem}

From the above theorems, it follows that our modified estimators still achieve $n^{-1/2}$ consistency in the mixed setup whenever the density $f$ is sufficiently smooth, i.e., $s>d/2$. Hence, in this case, we achieve the minimax optimal rate for the pure continuous setup \citep{birge1995estimation}, even in the presence of atoms in the data distribution. Therefore, our method achieves optimal statistical efficiency even in the presence of atoms in the data-generating distribution, demonstrating both robustness and sharpness of performance in the mixed setting.

\Cref{tab:simple} summarizes our main theoretical results alongside their counterparts from \cite{NIPS2015_06138bc5}. It is worth noting that if the data are indeed generated from a purely continuous distribution (i.e., without any atoms), then our estimators reduce exactly to the standard ones, and the corresponding performance guarantees remain unchanged.
 While our results (in the presence of atoms in the data distribution) are similar to theirs (without atoms, i.e., in a purely continuous setup), they offer additional flexibility by accommodating finite support of the discrete distribution that can grow with sample size in a triangular-array setup. Importantly, consistency still holds even when the discrete component has countably many atoms.

\begin{table}[!ht]
\centering
\caption{Comparison of our theoretical results with those in \cite{NIPS2015_06138bc5}. The assumption that $f$ (and/or $g$) lies in $\Sigma(s, L, B_0, B)$ is a common assumption in all the theoretical results. Assumptions except those are listed in the second column. MAE denotes mean absolute error, $\mathbb E|\hat{T}- T|$ and MSE denotes mean square error, $\mathbb E(\hat{T}- T)^2$.\\}
\label{tab:thm-comparison}
\label{tab:simple}
    \resizebox{\linewidth}{!}{
    \begin{tabular}{lccc}
    \toprule
    \addlinespace
Type & Assumptions &  {{Finitely many atoms (our results)}} &\specialcell{No atoms (\cite{NIPS2015_06138bc5}),\\special case of our method)} \\[3.3mm]
\hline

\multirow{2}{*}{\specialcell{DS \\(1 dist)}}& \eqref{eq:assmp-1d} & $\text{MAE}=\begin{cases}
    \mathcal{O}\left(n^{\frac{-2s}{2s+d}}\right),s<d/2\\
    \mathcal{O}(n^{-1/2}), s\geq d/2
\end{cases}$ &  $\text{MSE}=\begin{cases}
    \mathcal{O}\left(n^{\frac{-4s}{2s+d}}\right),s<d/2\\
    \mathcal{O}(n^{-1}), s\geq d/2
\end{cases}$
 \\
 \cline{2-4}
 & \specialcell{$s > d/2$, \\ $\psi\neq 0$,  \eqref{eq:assmp-1d}}
 & $\sqrt{n}$-asymptotic normality  & $\sqrt{n}$-asymptotic normality
 \\

\hline

\multirow{2}{*}{\specialcell{DS \\(2 dist)}} & \eqref{eq:assmp-2d-1}, \eqref{eq:assmp-2d-2} & $\text{MAE}=\begin{cases}
    \mathcal{O}\left(n^{\frac{-2s}{2s+d}}+m^{\frac{-2s}{2s+d}}\right),s<d/2\\
    \mathcal{O}(n^{-1/2}+m^{-1/2}), s\geq d/2
\end{cases}$ &  $\text{MSE}=\begin{cases}
    \mathcal{O}\left(n^{\frac{-4s}{2s+d}}+m^{\frac{-4s}{2s+d}}\right),s<d/2\\
    \mathcal{O}(n^{-1}+m^{-1}), s\geq d/2
\end{cases}$\\ 
\cline{2-4}
&\specialcell{$s > \frac{d}{2}$, \eqref{eq:assmp-2d-1}, \eqref{eq:assmp-2d-2},\\$\psi_f,\psi_g\neq 0$ }    & $\sqrt{n}$-asymptotic normality  & $\sqrt{n}$-asymptotic normality
 \\

\hline

\specialcell{LOO \\(1 dist)} &  \eqref{eq:assmp-1d} & $\text{MAE}=\begin{cases}
    \mathcal{O}\left(n^{\frac{-2s}{2s+d}}\right),s<d/2\\
    \mathcal{O}(n^{-1/2}), s\geq d/2
\end{cases}$ &  $\text{MSE}=\begin{cases}
    \mathcal{O}\left(n^{\frac{-4s}{2s+d}}\right),s<d/2\\
    \mathcal{O}(n^{-1}), s\geq d/2
\end{cases}$
 \\

\hline

\specialcell{LOO \\(2 dist)} & \eqref{eq:assmp-2d-1}, \eqref{eq:assmp-2d-2} & $\text{MAE}=\begin{cases}
    \mathcal{O}\left(n^{\frac{-2s}{2s+d}}+m^{\frac{-2s}{2s+d}}\right),s<d/2\\
    \mathcal{O}(n^{-1/2}+m^{-1/2}), s\geq d/2
\end{cases}$ &  $\text{MSE}=\begin{cases}
    \mathcal{O}\left(n^{\frac{-4s}{2s+d}}+m^{\frac{-4s}{2s+d}}\right),s<d/2\\
    \mathcal{O}(n^{-1}+m^{-1}), s\geq d/2
\end{cases}$
\\

\hline
\end{tabular}}
\end{table}

\section{Experiments}
We now present a series of simulation experiments demonstrating the practical advantages of our atom-aware methodology over standard estimators.

\label{sec:expt}
\subsection{Density estimator} We first evaluate the proposed method for estimating the density of a mixture of univariate continuous and discrete distributions. We consider data generated from the mixture model   $0.6 \mathcal{N}(0,1) + 0.4 \text{Binomial}(10, 0.5).$
The continuous component is estimated using a kernel density estimator with a Gaussian kernel and bandwidth 
$h = 1.06 \cdot \hat{\sigma} n^{-1/5}$,
where \( \hat{\sigma} \) is the empirical standard deviation of the \textit{unique} observed values. The discrete component is estimated by assigning a point mass at each repeated observed value with its relative frequency. \Cref{fig:1} shows estimated vs. true densities and PMFs for different sample size values $n=100,1000,10000$ along with the standard KDE approach (which naïvely applies a continuous estimator to the entire dataset), implemented using the \texttt{kde} function from the \texttt{ks} package in \textsf{R}.

Now, we present a similar experiment with multivariate continuous and discrete distributions. We consider i.i.d. observations generated from the mixture model $0.6 (X,Y) + 0.4 (Z,0)$, where $(X,Y)\sim\mathcal{N}_2(0,I_2)$ and $Z\sim \text{Pois}(1)$.
 \Cref{fig:kde-2d} shows our modified KDE, along with the standard KDE approach (which naïvely applies a continuous estimator to the entire dataset), implemented using the \texttt{kde} function from the \texttt{ks} package in \textsf{R}. The results highlight that the standard method fails to capture the structure of the mixture, whereas our straightforward modification leads to accurate estimation.

\begin{figure*}[!htb]
\centering
\centering
\subfloat[Usual KDE fails in the presence of atoms.]{\includegraphics[width=0.5\linewidth]{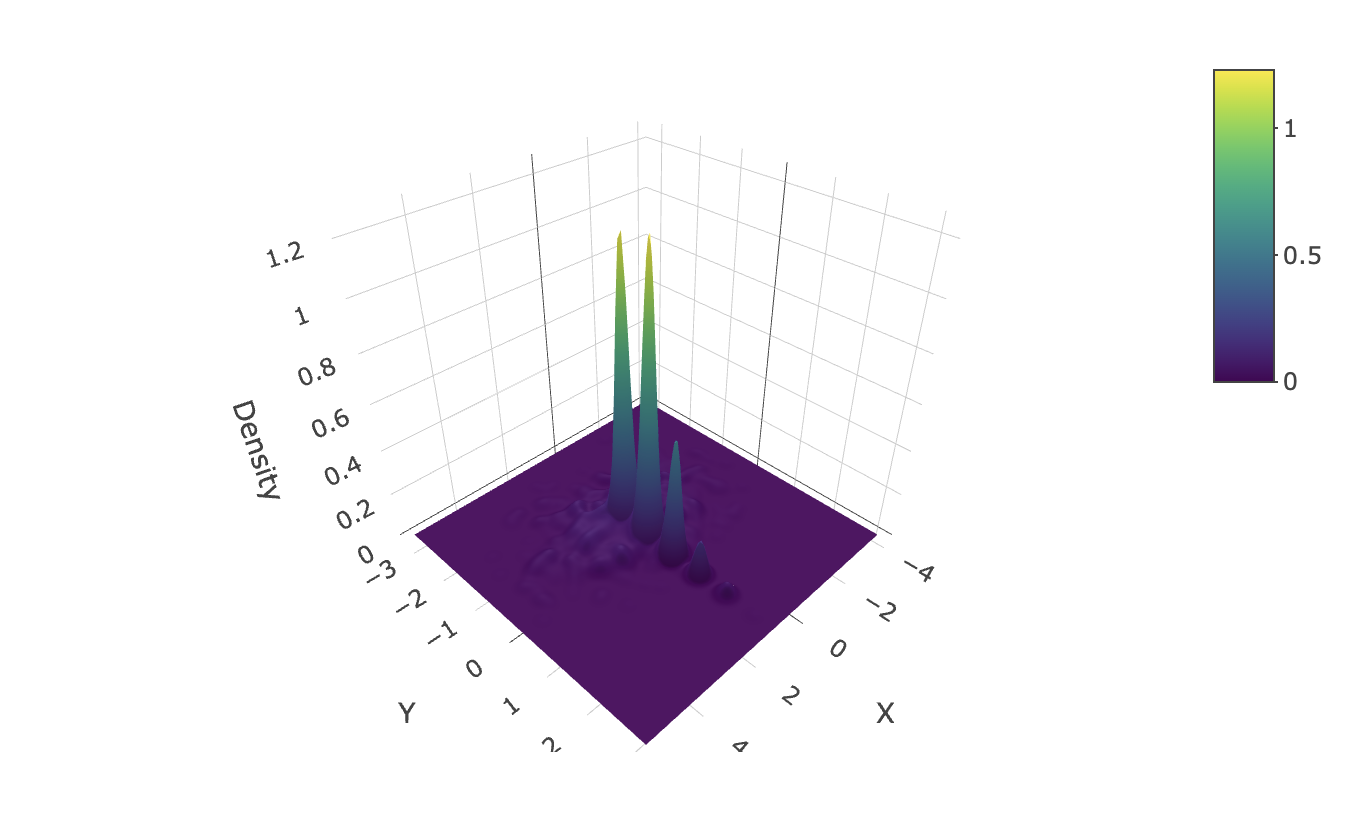}} 
\subfloat[Our simple modification works.]{\includegraphics[width=0.5\linewidth]{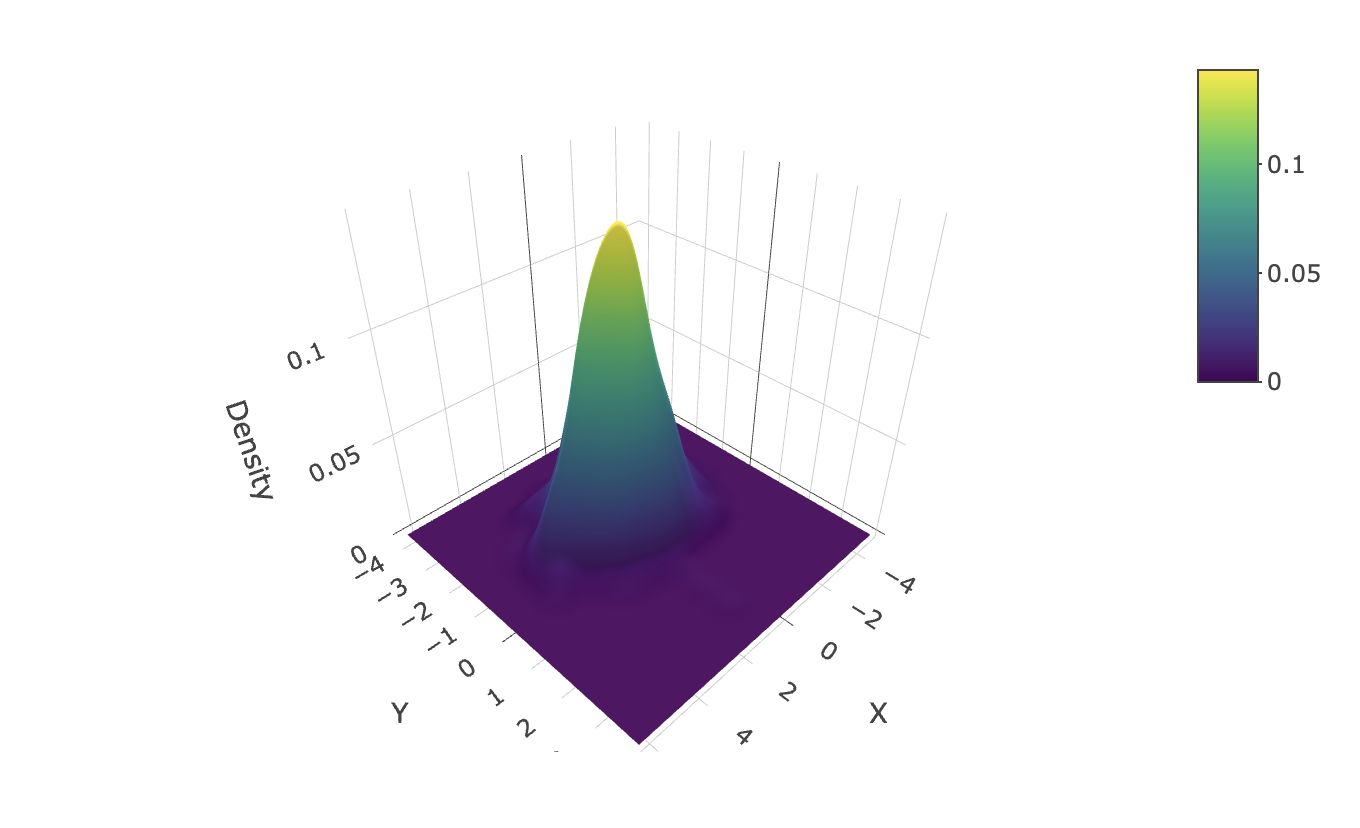}}
\caption[]{Density estimated using $n=1000$ samples drawn from the mixture $0.6 (X,Y) + 0.4 (Z,0)$, where $(X,Y)\sim\mathcal{N}_2(0,I_2)$ and $Z\sim \text{Pois}(1)$.} 
\label{fig:kde-2d}  
\end{figure*}

\subsection{Entropy estimator}
We generate i.i.d.\ samples from a mixture of a continuous and a discrete distribution, where the discrete component is a scaled Poisson distribution, $\mathrm{Poisson}(1)/5$ (supported on a countable set). For the continuous part, we consider two cases: (i) the uniform distribution on $[0,1]$, and (ii)  density $0.5 + 5t^5$ for $t \in [0,1]$. We then compute our leave-one-out (LOO) estimator $\hat{T}^{\text{LOO}}_{\mathcal{U}_n^1}$ for Shannon entropy under both settings. To assess performance, we report the average absolute error over $100$ independent runs and compare our method against two baselines: (a) the LOO estimator from \cite{NIPS2015_06138bc5}, which uses the full data and hence, inconsistent when atoms are present. (a) the oracle estimator: the estimator is the same, but it now has access to the labels indicating whether each point was generated from the continuous component, and uses only the continuous part for estimation. The results, shown in \Cref{fig:entropy}, highlight that the mean absolute error of our method is very close to that of the oracle.

\begin{figure*}[!htb]
\centering
\centering
\subfloat[$f_1(t)=1, 0\leq t\leq 1$]{\includegraphics[width=0.48\linewidth]{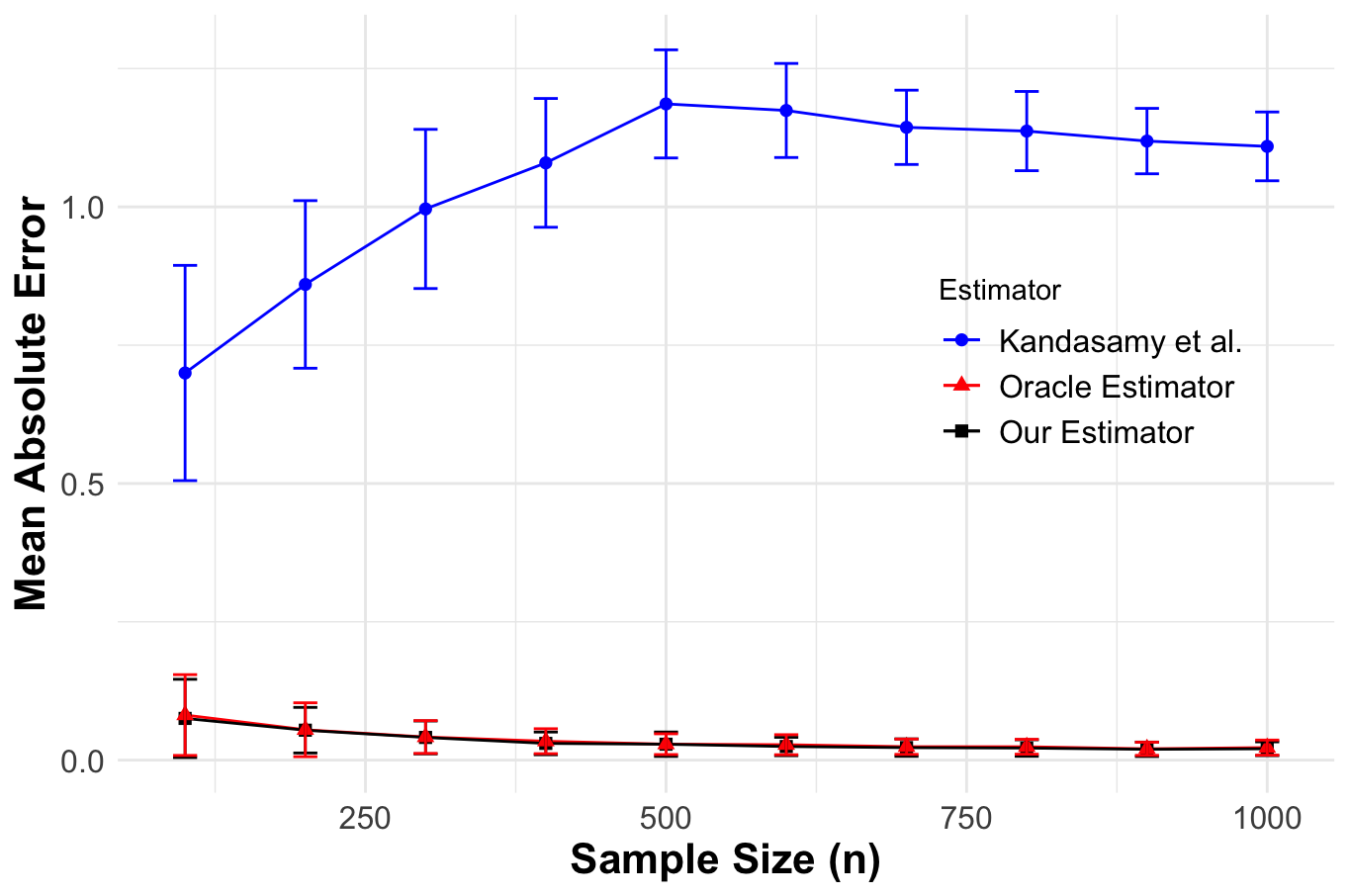}} 
\subfloat[$f_1(t)=0.5 + 5t^5, 0\leq t\leq 1$]{\includegraphics[width=0.48\linewidth]{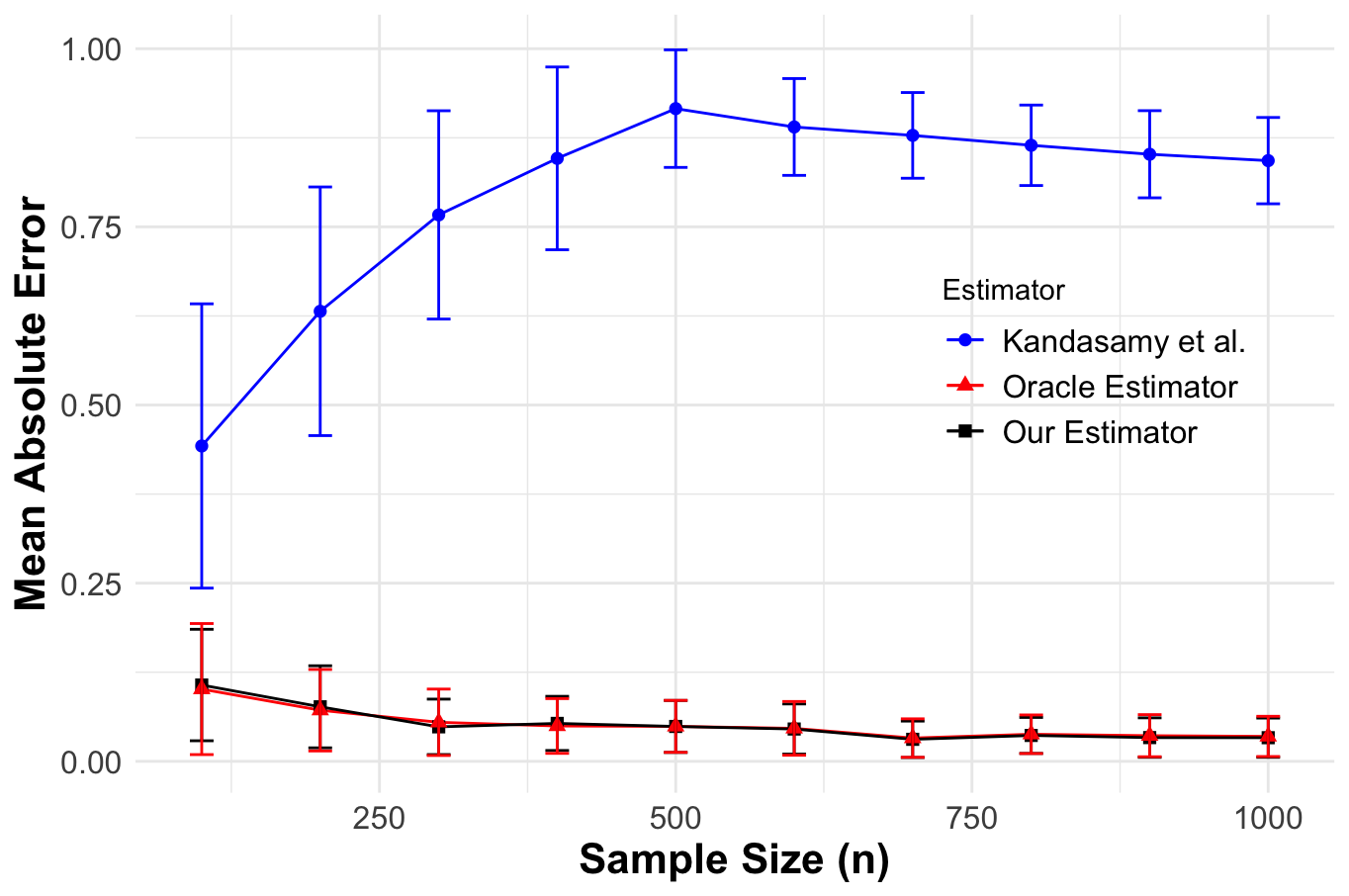}}
\caption[]{The average of the absolute error of entropy estimation is plotted against the sample size. Here, $60\%$ of the data is drawn from the density $f_1$ and the remaining $40\%$ from $\mathrm{Poisson}(1)/5$. Our atom-aware estimator closely matches the performance of the oracle that has access to the labels, and their mean absolute error approaches zero as the sample size increases. However, the original estimator of \cite{NIPS2015_06138bc5} fails due to its inability to handle atoms in the distribution.
} 
\label{fig:entropy}  
\end{figure*}

\section{Conclusion and future work}
\label{sec:conc}
We presented a simple yet powerful framework for nonparametric estimation of density and density functionals of the continuous component in the presence of mixed discrete-continuous data. By leveraging the empirical observation that unique values are likely drawn from the continuous component while repeated ones stem from the discrete part, our method cleanly separates the two components without prior knowledge of the discrete support. We showed that this simple idea integrates naturally with existing estimators and maintains their consistency and optimality under standard smoothness assumptions. Our theoretical results and empirical evaluations highlight the flexibility and effectiveness of our approach, opening the door to robust estimation in broader mixed-data scenarios. We use \cite{NIPS2015_06138bc5} just as a concrete testbed; the core idea is broadly applicable and not tied to any specific methodology.

Our work opens several avenues for further investigation. It might be interesting to extend our atom-aware methodology to other nonparametric density and functional estimation frameworks that currently assume purely continuous distributions. In particular, the family of ensemble estimators proposed by  \cite{moon2018ensemble} aggregates multiple plug-in KDE divergence estimators, and it is plausible that, with suitable modifications, they can be adapted to mixed discrete-continuous distributions while preserving consistency and optimality guarantees. One can also employ this idea for k-nearest neighbour (k-NN) density estimation and fixed-k-NN
density functional estimators \citep{singh2016finite}.
Another natural direction involves estimating functionals of the discrete component from a discrete-continuous mixture.

Finally, we note that many real datasets are inherently mixed in nature, and a broader range of consistent and efficient estimators that are robust to such heterogeneity could substantially enhance the reliability of modern data-driven applications.

\bibliographystyle{abbrvnat}
\bibliography{ref.bib}

\newpage
\appendix
\small
\allowdisplaybreaks
\pagenumbering{arabic}   
\setcounter{page}{1}     
\renewcommand\theequation{A.\arabic{equation}}
\setcounter{equation}{0}

\begin{center}
  \vspace*{0.5cm}
  {\LARGE\bfseries Supplementary Material:}\\[1em]
  {\large Density estimation with atoms, and functional
estimation for mixed discrete-continuous data}\\[1em]
\end{center}
\vspace{0.5cm}

\section{Mathematical details}

Note that, for $i=1,\cdots,n$, we can write $X_i\sim(1-\pi_1)F+\pi H_1$ as 
\begin{equation}
    X_i=(1-\Lambda_i)V_i+\Lambda_i U_i, ~~~\text{ where }\Lambda_i\sim \text{ Ber}(\pi_1), V_i\sim F, U_i\sim H_1
\end{equation}
and all $U_i,V_i,\Lambda_i$ are independent. Similarly, for $i=1,\cdots,m$, we have
\begin{equation}
    Y_i=(1-\Gamma_i)Z_i+\Gamma_i W_i, ~~~\text{ where }\Gamma_i\sim \text{ Ber}(\pi_2), Z_i\sim G, W_i\sim H_2
\end{equation}
and all $W_i,Z_i,\Gamma_i$ are independent.
\subsection{Auxiliary Lemmas}

\begin{lemma}
\label{lem:rn-to-s}
Let $S$ be the countable support of $H_1$. Then,
    $\mathcal{R}_n\uparrow\mathcal{S}$ almost surely as $n\to\infty$, i.e., $\P\left[\cup_{n=1}^\infty\mathcal{R}_n=\mathcal{S}\right]=1.$
\end{lemma}
\begin{proof}
   We observe that $\mathcal{R}_n\subseteq \mathcal{S}$ almost surely for all $n$ and hence, $\cup_{n=1}^\infty\mathcal{R}_n\subseteq \mathcal{S}$ almost surely.
   
  Now, for the sake of contradiction, suppose that $\P\left[\cup_{n=1}^\infty\mathcal{R}_n=\mathcal{S}\right]<1.$  Then, it follows that $\P\left[\cup_{n=1}^\infty\mathcal{R}_n\subsetneq\mathcal{S}\right]>0.$ So,
  \begin{equation*}
      \sum_{x\in\mathcal{S}}\P[x\notin \cup_{n=1}^\infty\mathcal{R}_n]\geq \P(\cup_{x\in\mathcal{S}}[x\notin \cup_{n=1}^\infty\mathcal{R}_n])=\P\left[\cup_{n=1}^\infty\mathcal{R}_n\subsetneq\mathcal{S}\right]>0.
  \end{equation*}
Hence, there exists some $x\in\mathcal{S}$ with $\P[x\notin \cup_{n=1}^\infty\mathcal{R}_n]>0$. But $x\in\mathcal{S}$ imples that $p_x>0$ and using SLLN, we have  $\P[x\notin \cup_{n=1}^\infty\mathcal{R}_n]=\P[\sum_{i=1}^\infty \mathds{1}(X_i=x)\leq 1]=0$, which is a contradiction. 
    Thus, we have shown $\P\left[\cup_{n=1}^\infty\mathcal{R}_n=\mathcal{S}\right]=1.$
\end{proof}

\begin{lemma}
\label{lem:rn-to-s-finite}
If $S$, the support of $H_1$, is finite, then
$\P\left[\cup_{n=1}^N\mathcal{R}_n=\mathcal{S} \text{ for all large enough } N\right]=1.$
\end{lemma}
\begin{proof}
    From \Cref{lem:rn-to-s}, there exists a null set $N$ such that $\P(N)=0$ and for all event $\omega\in N^c$, $\cup_{n=1}^\infty\mathcal{R}_n(\omega)=\mathcal{S}$.
    Since $S$ is finite, there exists $N_0(\omega)$ such that
    $\cup_{n=1}^N\mathcal{R}_n(\omega)=\mathcal{S}$, for all $N\geq N_0(\omega)$. Therefore, $\P\left[\cup_{n=1}^N\mathcal{R}_n=\mathcal{S} \text{ for all large enough } N\right]\geq \P(N^c)=1.$
\end{proof}

\begin{lemma}
\label{thm:pi-hat-convergence}
 $\frac{1}{n}\sum_{i=1}^n (\mathds{1}(X_i\in \mathcal{S})-\mathds{1}(X_i\in \mathcal{R}_n))\to 0$ almost surely as $n\to\infty$.    
\end{lemma}
\begin{proof}
It follows from \Cref{lem:rn-to-s} that $\exists$ a null set $N$, such that $\forall \omega \in N^c, \cup_{n=1}^\infty\mathcal{R}_{n}(\omega)=\mathcal{S}$. 
This implies that $\forall \omega \in N^c, (\mathcal{S}\setminus\mathcal{R}_n(\omega))\downarrow\emptyset$ as $n\to\infty$,

Fix any $\epsilon>0.$ So, there exists a finite set $\mathcal{S}^\prime\subseteq \mathcal{S}$ such that $\sum_{x\in\mathcal{S}^\prime}p_x\geq 1-\epsilon.$
Now, it follows from \Cref{lem:rn-to-s} that $\exists$ a null set $N$, such that $\forall \omega \in N^c, \cup_{n=1}^\infty\mathcal{R}_{n}(\omega)=\mathcal{S}$. 
This implies that $\forall \omega \in N^c,\exists n_0(\omega)$ such that $(\mathcal{S}\setminus\mathcal{R}_n(\omega))\subseteq (\mathcal{S}\setminus\mathcal{S}^\prime)$ for all $n\geq n_0(\omega)$.

Therefore, for $n\geq n_0(\omega)$,
$$\frac{1}{n}\sum_{i=1}^n (\mathds{1}(X_i\in \mathcal{S})-\mathds{1}(X_i\in \mathcal{R}_n(\omega)))=\frac{1}{n}\sum_{i=1}^n \mathds{1}(X_i\in \mathcal{S}\setminus\mathcal{R}_n(\omega))\leq \frac{1}{n}\sum_{i=1}^n \mathds{1}(X_i\in \mathcal{S}\setminus\mathcal{S}^\prime).$$ 
Now, by SLLN, $ \frac{1}{n}\sum_{i=1}^n \mathds{1}(X_i\in \mathcal{S}\setminus\mathcal{S}^\prime)\to \pi\sum_{x\in\mathcal{S}\setminus\mathcal{S}^\prime}p_x\leq \pi\epsilon,$
as $n\to \infty$. Hence, $$\limsup_n\frac{1}{n}\sum_{i=1}^n (\mathds{1}(X_i\in \mathcal{S})-\mathds{1}(X_i\in \mathcal{R}_n(\omega)))\leq \pi\epsilon$$ and $\epsilon$ can be made arbitrarily small, and so we have 
\begin{equation}
\label{eq:conv1}
    \frac{1}{n}\sum_{i=1}^n (\mathds{1}(X_i\in \mathcal{S})-\mathds{1}(X_i\in \mathcal{R}_n(\omega)))\to 0
\end{equation}
 as $n\to\infty.$ 
\end{proof}

Define the following random variable by replacing all $X_i$'s in $\hat{f}_{\mathcal{U}_n^1}$ with $V_i$'s:
\begin{equation}
    \hat{f}^V_{\mathcal{U}_n^1}(x)= \frac{1}{(|\mathcal{U}_n| \vee 1)h} \sum_{1\leq i\leq n:X_i \in \mathcal{U}_n^1} K\left( \frac{x - V_i}{h} \right)
\end{equation}

\begin{lemma}
The estimator \( \hat{f}_{\mathcal{U}_n^1}(x) \) satisfies
\[
  \int|\hat{f}^V_{\mathcal{U}_n^1}(x)-\hat{f}_{\mathcal{U}_n^1}(x)|dx\stackrel{a.s.}{\leq}\frac{2}{|\mathcal{U}_n|\vee 1}\sum_{{i=1}}^n\mathds{1}(X_i \in \mathcal{S}\setminus\mathcal{R}_n).  \]
\end{lemma}

\begin{proof}

\begin{align*}
&\int|\hat{f}^V_{\mathcal{U}_n^1}(x)-\hat{f}_{\mathcal{U}_n^1}(x)|dx\\
&\leq\frac{1}{(|\mathcal{U}_n| \vee 1) h}\left[\sum_{\substack{1\leq i\leq n:\\X_i \in \mathcal{U}_n\cap\mathcal{S}}}\int\left| K\left( \frac{x - V_i}{h} \right)-K\left( \frac{x - X_i}{h} \right)\right|dx+\sum_{\substack{1\leq i\leq n:\\X_i \in \mathcal{U}_n\cap\mathcal{S}^c}}\int\left| K\left( \frac{x - V_i}{h} \right)-K\left( \frac{x - X_i}{h} \right)\right|dx\right]\\
&=\frac{1}{|\mathcal{U}_n| \vee 1}\sum_{\substack{1\leq i\leq n:\\X_i \in \mathcal{U}_n\cap\mathcal{S}}}\int\left| K\left(z \right)-K\left(z+ \frac{V_i - X_i}{h} \right)\right|dz+\frac{1}{(|\mathcal{U}_n| \vee 1) h}\sum_{\substack{1\leq i\leq n:\\X_i \in \mathcal{U}_n\cap\mathcal{S}^c}}\int\left| K\left( \frac{x - V_i}{h} \right)-K\left( \frac{x - V_i}{h} \right)\right|dx\\
&\leq \frac{2}{|\mathcal{U}_n| \vee 1}\sum_{\substack{1\leq i\leq n}}\mathds{1}(X_i \in \mathcal{U}_n\cap\mathcal{S})+0\\
&\leq \frac{2n}{(|\mathcal{U}_n| \vee 1)}\frac{1}{n}\sum_{{i=1}}^n\mathds{1}(X_i \in \mathcal{S}\setminus\mathcal{R}_n)
\end{align*}
where the second inequality follows from the fact that $$\int\left| K\left(z \right)-K\left(z+ \frac{V_i - X_i}{h} \right)\right|dz\leq \int K\left(z \right)dz+\int K\left(z+ \frac{V_i - X_i}{h} \right)dz=2$$ and last inequality follows from the fact that $\mathcal{U}_n\cap\mathcal{S}\subseteq\mathcal{S}\setminus\mathcal{R}_n$. The second term is almost surely $0$, because on $X_i\in\mathcal{S}^c\implies X_i=V_i$ almost surely. Now,  the first term converges to $0$ almost surely, because $\frac{1}{n}\sum_{{i=1}}^n\mathds{1}(X_i \in \mathcal{S}\setminus\mathcal{R}_n)\to 0$ (which follows from \eqref{eq:conv1}) and $\frac{n}{|\mathcal{U}_n|}\to 1/(1-\pi)$ almost surely, as $n\to\infty$.

Thus, we have 
\begin{equation}
\label{eq:f-hat-diff}
\int|\hat{f}^V_{\mathcal{U}_n^1}(x)-\hat{f}_{\mathcal{U}_n^1}(x)|dx\stackrel{a.s.}{\leq}\frac{2n}{|\mathcal{U}_n|\vee 1}\frac{1}{n}\sum_{{i=1}}^n\mathds{1}(X_i \in \mathcal{S}\setminus\mathcal{R}_n).
\end{equation}
\end{proof}

\begin{lemma}
\label{lem:exp-decay}
As $n\to\infty$,    $\mathbb E[\mathds{1}(X_1 \in \mathcal{S}\setminus\mathcal{R}_n)]\to 0.$ Moreover, if $\mathcal{S}$, the support of $H_1$ is finite, then $\mathbb E[\mathds{1}(X_1 \in \mathcal{S}\setminus\mathcal{R}_n)]=\mathcal{O}(1/\kappa^n),$ for some constant  $\kappa=1/(1-\pi_1 \min_{s\in \mathcal{S}}p_s)$, which depends only on $H_1$ and $\pi_1$. Further, consider the triangular array setup, where $X_{n,1},\cdots,X_{n,n}\stackrel{iid}{\sim}(1-\pi_1)F_1+\pi_1H_{1,n}$ and $H_{1,n}$ has finite support $\mathcal{S}_n$, which may grow with $n$, and let its probability mass function (p.m.f.) be $\{p^{(n)}_s\}_{s\in \mathcal{S}_n}$. Assume that the minimum mass of an atom satisfies $\min_{s\in \mathcal{S}_n}p_s^{(n)}\geq\frac{1}{\pi_1}(1-(cn^{-\frac{ks}{2s+d}})^\frac{1}{n-1})$, for some constant $c>0$. Then, $\mathbb E[\mathds{1}(X_{n,1} \in \mathcal{S}_n\setminus\mathcal{R}_n)]=\mathcal{O}(n^{-\frac{ks}{2s+d}}),$ for any $k>0.$ 

\end{lemma}

\begin{proof}
    \begin{align*}
        \mathbb E[\mathds{1}(X_1 \in \mathcal{S}\setminus\mathcal{R}_n)]&=\mathbb P[X_1 \in \mathcal{S}, X_j\neq X_1, \text{ for }j=2,\cdots,n]\\
        &=\sum_{s\in \mathcal{S}}\mathbb P[X_1 =s, X_j\neq s, \text{ for }j=2,\cdots,n]\\
        &=\sum_{s\in \mathcal{S}}\pi_1 p_s(1-\pi_1 p_s)^{n-1}.
    \end{align*}
    Fix any $\epsilon>0$. There exists a finite set $\mathcal{S}_1\subseteq \mathcal{S}$ such that $\sum_{s\in \mathcal{S}_1} p_s\geq 1-\frac{\epsilon}{2\pi_1}.$ Therefore, $$ \mathbb E[\mathds{1}(X_1 \in \mathcal{S}\setminus\mathcal{R}_n)]\leq \pi_1(1-\pi_1 \min_{s\in \mathcal{S}_1}p_s)^{n-1}+\epsilon/2.$$

    We can choose $n$ large enough so that $\pi_1(1-\pi_1 \max_{s\in \mathcal{S}_1}p_s)^{n-1}\leq\epsilon/2$ and hence, $\mathbb E[\mathds{1}(X_1 \in \mathcal{S}\setminus\mathcal{R}_n)]\leq\epsilon$, for all large enough $n.$ Since $\epsilon$ can be arbitrarily small, we have
    \begin{equation}
        \mathbb E[\mathds{1}(X_1 \in \mathcal{S}\setminus\mathcal{R}_n)]\to 0, \text{ as }n\to \infty.
    \end{equation}

    Now, if $\mathcal{S}$ is finite, 
    \begin{align*}
        \mathbb E[\mathds{1}(X_1 \in \mathcal{S}\setminus\mathcal{R}_n)
        &=\sum_{s\in \mathcal{S}}\pi_1 p_s(1-\pi_1 p_s)^{n-1}\leq \pi_1(1-\pi_1 \min_{s\in \mathcal{S}}p_s)^{n-1}.
    \end{align*}
    Choose $\kappa=1/(1-\pi_1 \min_{s\in \mathcal{S}}p_s)$ to obtain $\mathbb E[\mathds{1}(X_1 \in \mathcal{S}\setminus\mathcal{R}_n)]=\mathcal{O}(1/\kappa^n).$

For the last part, $$ \mathbb E[\mathds{1}(X_{n,1} \in \mathcal{S}_n\setminus\mathcal{R}_n)]=\sum_{s\in \mathcal{S}_n}\pi_1 p_s(1-\pi_1 p_s)^{n-1}\leq \pi_1(1-\pi_1 \min_{s\in \mathcal{S}_n}p_s)^{n-1}\leq \pi_1cn^{-\frac{ks}{2s+d}}.$$

\end{proof}

\subsection{Proofs of theorems stated in the main paper}

\subsubsection{Proof of \Cref{thm:kde}}

\begin{proof}
 Define, $A_n=\{k\in0,1,\cdots,n-1: |k-n\pi|<n^{2/3}\}$. Also note that $n-|\mathcal{U}_n^1|\leq\sum_{i=1}^n \Lambda_i$ almost surely. 
Now, we will show that $\lim\sup_{n\to\infty}\mathbb E\left[\frac{n^2}{(|\mathcal{U}_n^1| \vee 1)^2}\right]<\infty$.
\begin{align*}
\mathbb E\left[\frac{n^2}{(|\mathcal{U}_n^1| \vee 1)^2}\right]&\leq \mathbb E\left[\frac{n^2}{((n-\sum_{i=1}^n \Lambda_i) \vee 1)^2}\right]\\
&\leq \sum_{k=1}^{n-1}\frac{n^2}{(n-k)^2}\P\left[\sum_{i=1}^n \Lambda_i=k\right]\\
&\leq\sum_{k\in A_n}\frac{n^2}{(n-n\pi-n^{2/3})^2}\P\left[\sum_{i=1}^n \Lambda_i=k\right]+\sum_{k\in A_n^c}n^2\P\left[\sum_{i=1}^n \Lambda_i=k\right]\\
&\leq\frac{n^2}{(n-n\pi-n^{2/3})^2}+n^2\P\left[\left|\sum_{i=1}^n \Lambda_i-n\pi\right|\geq n^{2/3}\right]\\
&\leq\frac{n^2}{(n-n\pi-n^{2/3})^2}+n^2\exp\{-2n^{1/3}\}\to 1/(1-\pi)^2, \text{ as } n\to \infty.
\end{align*}
The last inequality above follows from Hoeffding's Inequality for Bernoulli random variables.
From \eqref{eq:f-hat-diff}, and applying the Cauchy-Schwarz inequality, we get
\begin{equation}
   \mathbb E\left[\int|\hat{f}^V_{\mathcal{U}_n^1}(x)-\hat{f}_{\mathcal{U}_n^1}(x)|dx\right]\leq\mathbb E\left[\frac{2n \mathds{1}(X_1 \in \mathcal{S}\setminus\mathcal{R}_n)}{|\mathcal{U}_n|\vee 1}\right]\leq 2\sqrt{E\left[\frac{n^2}{(|\mathcal{U}_n|\vee 1)^2}\right]\mathbb E[\mathds{1}(X_1 \in \mathcal{S}\setminus\mathcal{R}_n)]}.
\end{equation}
Now, it follows from \Cref{lem:exp-decay} that $\mathbb E\left[\int|\hat{f}^V_{\mathcal{U}_n^1}(x)-\hat{f}_{\mathcal{U}_n^1}(x)|dx\right]\to 0,$ as $n\to\infty$ and if $H_1$ has finite support $\mathcal{S}_n$, 
 using the last part of \Cref{lem:exp-decay}, we obtain from above that 
 \begin{equation}
 \label{eq:expt-kde-diff}
     \mathbb E\left[\int|\hat{f}^V_{\mathcal{U}_n^1}(x)-\hat{f}_{\mathcal{U}_n^1}(x)|dx\right] \leq 2c\sqrt{\pi_1E\left[\frac{n^2}{(|\mathcal{U}_n|\vee 1)^2}\right]}\times n^{-\frac{s}{2s+d}}=\mathcal{O}(n^{-\frac{s}{2s+d}}).
 \end{equation}
 Therefore, it is enough to show that 
 $\mathbb E\left[\int|\hat{f}^V_{\mathcal{U}_n^1}(x)-{f}(x)|dx\right]=\mathcal{O}(n^{-\frac{s}{2s+d}}).$ For that, we will use a standard KDE result \cite{devroye1985nonparametric} that under the assumptions of the theorem, $\mathbb E\left[\int|\hat{f}^V_{k}(x)-{f}(x)|^2 dx\right]=\mathcal{O}(n^{-\frac{2s}{2s+d}}).$ 
 \begin{align*}
    &\mathbb E\left[\int|\hat{f}^V_{\mathcal{U}_n^1}(x)-{f}(x)|^2dx\right]\\
    &=\sum_{k=1}^\infty \mathbb E\left(\int|\hat{f}^V_{\mathcal{U}_n^1}(x)-{f}(x)|^2dx\times \mathds{1}(|\mathcal{U}_n^1|=k)\right)\\
    &=\sum_{k\leq n(1-\pi)/2} \mathbb E\left[\int|\hat{f}^V_{k}(x)-{f}(x)|^2dx\right] \times \P(|\mathcal{U}_n^1|=k)+\sum_{k>n(1-\pi)/2} \mathbb E\left[\int|\hat{f}^V_{k}(x)-{f}(x)|^2dx\right] \times \P(|\mathcal{U}_n^1|=k)\\
    &\leq\sup_{k}\mathbb E\left[\int|\hat{f}^V_{k}(x)-{f}(x)|^2dx\right]\sum_{k\leq n(1-\pi)/2}  \P(|\mathcal{U}_n^1|=k)+\mathcal{O}(n^{-\frac{2s}{2s+d}})\times\sum_{k>n(1-\pi)/2} \P(|\mathcal{U}_n^1|=k)\\
    &\leq\sup_{k}\mathbb E\left[\int|\hat{f}^V_{k}(x)-{f}(x)|^2dx\right]  \P(n-\sum_{i=1}^n \Lambda_i\leq n(1-\pi)/2)+\mathcal{O}(n^{-\frac{2s}{2s+d}})\\
    &\leq\sup_{k}\mathbb E\left[\int|\hat{f}^V_{k}(x)-{f}(x)|^2dx\right]  \P(\sum_{i=1}^n \Lambda_i-n\pi\geq n(1-\pi)/2)+\mathcal{O}(n^{-\frac{2s}{2s+d}})\\
    &\leq\sup_{k}\mathbb E\left[\int|\hat{f}^V_{k}(x)-{f}(x)|^2dx\right]  \exp(-n(1-\pi)^2/2)+\mathcal{O}(n^{-\frac{2s}{2s+d}})~~[\text{By Hoeffding bound}]
\end{align*}
Since $\mathbb E\left[\int|\hat{f}^V_{k}(x)-{f}(x)|dx\right] \to 0$ as $n\to \infty$, we have $\sup_{k}\mathbb E\left[\int|\hat{f}^V_{k}(x)-{f}(x)|dx\right] <\infty$ and $\exp(-n(1-\pi)^2/2)\leq \mathcal{O}(n^{-\frac{s}{2s+d}})$.
Therefore,
\begin{equation}
\label{eq:kde-rate}
    \mathbb E\left[\int|\hat{f}^V_{\mathcal{U}_n^1}(x)-{f}(x)|^2dx\right]=\mathcal{O}(n^{-\frac{2s}{2s+d}}).
\end{equation}
Now, using Cauchy-Schwarz and Jensen's inequality, $\mathbb E\left[\int|\hat{f}^V_{{\mathcal{U}_n^1}}(x)-{f}(x)| dx\right]\leq \mathbb E\left[\sqrt{\int|\hat{f}^V_{{\mathcal{U}_n^1}}(x)-{f}(x)|^2 dx}\right]\leq\sqrt{\mathbb E\left[\int|\hat{f}^V_{{\mathcal{U}_n^1}}(x)-{f}(x)|^2 dx\right]}=\mathcal{O}(n^{-\frac{s}{2s+d}}).$
\end{proof}

\subsubsection{Proof of \Cref{thm:conv-ds}}

\begin{proof}

Define, 
\begin{equation}
\hat{T}^{\text{DS},1}_{\mathcal{U}_n^1,V} = T(\hat{f}_{\mathcal{U}_{n}^{1,1},V})+\frac{1}{|\mathcal{U}_{n}^{1,2}| \vee 1} \sum_{i:X_i\in\mathcal{U}_{n}^{1,2}}\psi(V_i;\hat{f}_{\mathcal{U}_{n}^{1,1},V})
\end{equation}
 Under the same assumptions, from Theorem 6 or 13 of \cite{NIPS2015_06138bc5},
we have that 
\begin{align*}
    \mathbb E|\hat{T}^{\text{DS}}_{n} - T(F)|^2=\mathcal{O}(n^{-\frac{4s}{2s+d}}+n^{-1})~~\text{ and } \sqrt{n}(\hat{T}^{\text{DS}}_{n}-T(F))\xrightarrow{d} N(0,\mathbb V_f(\psi(X,f))) , \text{ as } n \to \infty.
\end{align*} 
  We write 
\begin{align*}
\hat{T}^{\text{DS},1}_{\mathcal{U}_n^1,V}-\hat{T}^{\text{DS},1}_{\mathcal{U}_n^1}
&=T(\hat{f}_{\mathcal{U}_{n}^{1,1},V})-T(\hat{f}_{\mathcal{U}_{n}^{1,1}})+\frac{1}{|\mathcal{U}_{n}^{1,2}| \vee 1}\Bigg[\sum_{i:X_i \in \mathcal{U}_{n}^{1,2}\cap\mathcal{S}}\left( \psi(V_i;\hat{f}_{\mathcal{U}_{n}^{1,1},V})-\psi(X_i;\hat{f}_{\mathcal{U}_{n}^{1,1}})\right)\\
&+\sum_{i:X_i \in \mathcal{U}_{n}^{1,2}\cap\mathcal{S}^c}\left( \psi(V_i;\hat{f}_{\mathcal{U}_{n}^{1,1},V})-\psi(X_i;\hat{f}_{\mathcal{U}_{n}^{1,1}})\right)\Bigg]
\end{align*}
Note that 
\begin{align*}
|T(\hat{f}_{\mathcal{U}_{n}^{1,1},V})-T(\hat{f}_{\mathcal{U}_{n}^{1,1}})|&\leq L_\phi L_\nu\int \left|\hat{f}_{\mathcal{U}_{n}^{1,1},V}(x)-\hat{f}_{\mathcal{U}_{n}^{1,1}}(x)\right|dx \text{ as } n\to \infty,
\end{align*}
where $L_\phi$ and $L_\nu$ are the Lipschitz constants for the functions $\phi$ and $\nu$ respectively.
From the same steps as in the proof of Theorem 2.2, it follows that the first term
\begin{align*}
\mathbb E|T(\hat{f}_{\mathcal{U}_{n}^{1,1},V})-T(\hat{f}_{\mathcal{U}_{n}^{1,1}})|&\leq L_\phi L_\nu\mathbb E\left[\int \left|\hat{f}_{\mathcal{U}_{n}^{1,1},V}(x)-\hat{f}_{\mathcal{U}_{n}^{1,1}}(x)\right|dx \right]\to 0 \text{ as } n\to \infty,
\end{align*}
and is $\mathcal{O}(n^{-\frac{3s}{2s+d}})$, when $H_1$ has finite support $\mathcal{S}_n$ satisfying the given condition.

For the second term:
  \begin{align*}
&\mathbb E\left(\frac{1}{|\mathcal{U}_{n}^{1,2}| \vee 1}\sum_{i:X_i \in \mathcal{U}_{n}^{1,2}\cap\mathcal{S}}\left|\psi(V_i;\hat{f}_{{\mathcal{U}_{n}^{1,1},V}})-\psi(X_i;\hat{f}_{\mathcal{U}_{n}^{1,1}})\right|\right)\\
&\leq \mathbb E\left(\frac{2\|\psi\|_\infty}{|\mathcal{U}_{n}^{1,2}| \vee 1}\sum_{i=n/2}^n\mathds{1}(X_i \in \mathcal{U}_n\cap\mathcal{S})\right)\\
&\leq n\|\psi\|_\infty\times\mathbb E\left(\frac{\mathds{1}(X_i \in \mathcal{S}\setminus\mathcal{R}_n)}{|\mathcal{U}_{n}^{1,2}| \vee 1}\right)\\
&\leq \|\psi\|_\infty\sqrt{\mathbb E\left[\frac{n^2}{(|\mathcal{U}_{n}^{1,2}| \vee 1)^2}\right]\mathbb E[\mathds{1}(X_i \in \mathcal{S}\setminus\mathcal{R}_n)]}.
\end{align*}
From the similar steps as in the proof of Theorem 2.2, it follows that $\lim\sup_{n\to\infty}\mathbb E\left[\frac{n^2}{(|\mathcal{U}_{n}^{1,2}| \vee 1)^2}\right]<\infty$ and from \Cref{lem:exp-decay}, we have that the above converges to $0$ and is $\mathcal{O}(n^{-\frac{3s}{2s+d}})$, when $H_1$ has finite support $\mathcal{S}_n$ satisfying the given condition. Finally, for the last term,
\begin{align*}
\mathbb E\left|\frac{1}{|\mathcal{U}_{n}^{1,2}| \vee 1}\sum_{\substack{n/2\leq i\leq n:\\X_i \in \mathcal{U}_n\cap\mathcal{S}^c}}\left( \psi(V_i;\hat{f}_{\mathcal{U}_{n}^{1,1},V})-\psi(X_i;\hat{f}_{\mathcal{U}_{n}^{1,1}})\right)\right|&
  {\leq}\sum_{\substack{n/2\leq i\leq n}}\mathbb E\left| \frac{\psi(X_i;\hat{f}_{\mathcal{U}_{n}^{1,1},V})-\psi(X_i;\hat{f}_{\mathcal{U}_{n}^{1,1}})}{|\mathcal{U}_{n}^{1,2}| \vee 1}\right|\\
&=\frac{n}{2}\mathbb E\left| \frac{\psi(X_n;\hat{f}_{\mathcal{U}_{n}^{1,1},V})-\psi(X_n;\hat{f}_{\mathcal{U}_{n}^{1,1}})}{|\mathcal{U}_{n}^{1,2}| \vee 1}\right|,
\end{align*}
because if $X_i \in \mathcal{U}_n\cap\mathcal{S}^c,$ then $X_i=V_i$ almost surely, and by Assumption 4.1, for large enough $n$, $\mathbb E\left(\left|\psi(X_n;\hat{f}_{\mathcal{U}_n^{1,1},V})-\psi(X_n;\hat{f}_{\mathcal{U}_n^{1,1}}) \right|\Bigg\vert  X_{-n},V_{-n}\right)\leq C|\hat{f}_{\mathcal{U}_n^{1,1},V}-\hat{f}_{\mathcal{U}_n^{1,1}}|$, for some constant $C.$ So,
\begin{align*}
&n\mathbb E\left|\frac{\psi(X_n;\hat{f}_{\mathcal{U}_n^{1,1},V})-\psi(X_n;\hat{f}_{\mathcal{U}_n^{1,1}})}{|\mathcal{U}_n^{1,2}| \vee 1} \right|\\
&\leq n\mathbb E\left( \mathbb E\left(\left|\frac{\psi(X_n;\hat{f}_{\mathcal{U}_n^{1,1},V})-\psi(X_n;\hat{f}_{\mathcal{U}_n^{1,1}})}{|\mathcal{U}_n^{1,2}| \vee 1} \right|\Bigg\vert X_{-n},V_{-n}\right)\right) \\
&\leq \mathbb E\left(\frac{n}{(|\mathcal{U}_{n-1}^{1,2}|-1) \vee 1} \mathbb E\left(\left|\psi(X_n;\hat{f}_{\mathcal{U}_n^{1,1},V})-\psi(X_n;\hat{f}_{\mathcal{U}_n^{1,1}}) \right|\Bigg\vert  X_{-n},V_{-n}\right)\right) \\
&\leq \mathbb E\left( \frac{Cn}{(|\mathcal{U}_{n-1}^{1,2}|-1) \vee 1}\int |\hat{f}_{\mathcal{U}_n^{1,1},V}-\hat{f}_{\mathcal{U}_n^{1,1}}|\right)\\
&\leq \mathbb E\left[\frac{Cn}{(|\mathcal{U}_{n-1}^{1,2}|-1) \vee 1}\times\frac{2}{|\mathcal{U}_n^{1,2}| \vee 1}\sum_{{j\leq n/2}}\mathds{1}(X_j \in \mathcal{S}\setminus\mathcal{R}_n)\right]\\
&=Cn^2\mathbb E\left[\frac{1}{(|\mathcal{U}_{n-1}^{1,2}|-1) \vee 1}\frac{\mathds{1}(X_1 \in \mathcal{S}\setminus\mathcal{R}_n)}{|\mathcal{U}_n^{1,2}| \vee 1}\right]\\
&\leq C\sqrt[3]{\mathbb E\left[\frac{n^3}{((|\mathcal{U}_{n-1}^{1,2}| -1)\vee 1)^3}\right]\mathbb E\left[\frac{n^3}{(|\mathcal{U}_{n}^{1,2}| \vee 1)^3}\right]\mathbb E[\mathds{1}(X_1 \in \mathcal{S}\setminus\mathcal{R}_n)]} ~ ~ ,
\end{align*}
where $|\mathcal{U}_{n-1}^{1,2}|:=\{X_i:i>\floor{n/2},X_i\in \mathcal{U}_{n-1}^{1}\}$ and $\mathcal{U}_{n-1}^{1}$ denotes the number of unique elements in $X_{-n}=\{X_1,\cdots,X_{n-1}\}$, and the second inequality follows from the observation that conditioned on $X_{-n}$, $|\mathcal{U}_{n}^{1,2}|\geq |\mathcal{U}_{n-1}^{1,2}|-1$.
From the similar steps as in the proof of Theorem 2.2, it follows that $\lim\sup_{n\to\infty}\mathbb E\left[\frac{n^3}{((|\mathcal{U}_{n-1}^{1,2}| -1)\vee 1)^3}\right]<\infty$ and $\lim\sup_{n\to\infty}\mathbb E\left[\frac{n^3}{(|\mathcal{U}_{n}^{1,2}| \vee 1)^3}\right]<\infty$ and then, from \Cref{lem:exp-decay}, we have that the above converges to $0$ and is $\mathcal{O}(n^{-\frac{2s}{2s+d}})$, when $H_1$ has finite support $\mathcal{S}_n$ satisfying the given condition. 

Combining these three parts, we have 
  $\mathbb E|\hat{T}^{\text{DS},1}_{\mathcal{U}_n^1,V}-\hat{T}^{\text{DS},1}_{\mathcal{U}_n^1}|\to 0$, as $n\to\infty$ and when $H_1$ has finite support $\mathcal{S}_n$, we obtain $\mathbb E|\hat{T}^{\text{DS},1}_{\mathcal{U}_n^1,V}-\hat{T}^{\text{DS},1}_{\mathcal{U}_n^1}|=\mathcal{O}(n^{-\frac{2s}{2s+d}})$. Now, if we can show that  $\mathbb E|\hat{T}^{\text{DS},1}_{\mathcal{U}_n^1,V}-T(F)|=\mathcal{O}(n^{-\frac{2s}{2s+d}}+n^{-\frac{1}{2}})$, we are done with the $L_1$ convergence part. Since $\{V_i:1\leq i\leq n, X_i\in\mathcal{U}_{n}^{1,1}\}$ are i.i.d. having Lebesgue density $f$ (because $V_i$ and $[X_i\in\mathcal{U}_{n}^{1,1}]$ are independent) and $|\mathcal{U}_n^{1,1}|\leq n/2-\sum_{i=1}^{n/2} \Lambda_i$, as $n\to\infty$,
\begin{align*}
&\mathbb E|\hat{T}^{\text{DS},1}_{\mathcal{U}_n^1,V} - T(F)|^2\\
&=\sum_{k=1}^\infty \mathbb E\left(|\hat{T}^{\text{DS},1}_{\mathcal{U}_n^1,V} - T(F)|^2\mathds{1}(|\mathcal{U}_n^1|=k)\right)\\
&=\sum_{k\leq n(1-\pi)/2} \mathbb E|\hat{T}^{\text{DS},1}_{k} - T(F)|^2 \times \P(|\mathcal{U}_n^1|=k)+\sum_{k>n(1-\pi)/2} \mathbb E|\hat{T}^{\text{DS},1}_{k} - T(F)|^2 \times \P(|\mathcal{U}_n^1|=k)\\
&\leq\sup_{k}\mathbb E|\hat{T}^{\text{DS},1}_{k} - T(F)|^2 \sum_{k\leq n(1-\pi)/2}  \P(|\mathcal{U}_n^1|=k)+\mathcal{O}(n^{-\frac{4s}{2s+d}}+n^{-1})\times\sum_{k>n(1-\pi)/2} \P(|\mathcal{U}_n^1|=k)\\
&\leq\sup_{k}\mathbb E|\hat{T}^{\text{DS},1}_{k} - T(F)|^2   \P(n-\sum_{i=1}^n \Lambda_i\leq n(1-\pi)/2)+\mathcal{O}(n^{-\frac{4s}{2s+d}}+n^{-1})\\
&\leq\sup_{k}\mathbb E|\hat{T}^{\text{DS},1}_{k} - T(F)|^2   \P(\sum_{i=1}^n \Lambda_i-n\pi\geq n(1-\pi)/2)+\mathcal{O}(n^{-\frac{4s}{2s+d}}+n^{-1})\\
&\leq\sup_{k}\mathbb E|\hat{T}^{\text{DS},1}_{k} - T(F)|^2  \exp(-n(1-\pi)^2/2)+\mathcal{O}(n^{-\frac{4s}{2s+d}}+n^{-1}) . ~~~[\text{By Hoeffding bound}]
\end{align*}
Since $\mathbb E|\hat{T}^{\text{DS},1}_{k} - T(F)|^2 \to 0$ as $n\to \infty$, we have $\sup_{k}\mathbb E|\hat{T}^{\text{DS},1}_{k} - T(F)|^2 <\infty$ and $\exp(-n(1-\pi)^2/2)\leq \mathcal{O}(n^{-\frac{4s}{2s+d}}+n^{-1})$. The same results hold when $\hat{T}^{\text{DS},1}_{\mathcal{U}_n^1,V}$ and $\hat{T}^{\text{DS},1}_{\mathcal{U}_n^1}$ are replaced by $\hat{T}^{\text{DS},2}_{\mathcal{U}_n^1,V}$ and $\hat{T}^{\text{DS},2}_{\mathcal{U}_n^1}$ respectively. Thus, for $\hat{T}^{\text{DS}}_{\mathcal{U}_n^1}=(\hat{T}^{\text{DS},1}_{\mathcal{U}_n^1,V}+\hat{T}^{\text{DS},2}_{\mathcal{U}_n^1,V})/{2}$,  we conclude
\begin{equation}
 \mathbb E |\hat{T}^{\text{DS}}_{\mathcal{U}_n^1}-T(F)|   \to 0, \text{ as } n\to \infty
\end{equation}
and when $H_1$ has finite support $\mathcal{S}_n$ satisfying the given condition, 
\begin{equation}
    \mathbb E |\hat{T}^{\text{DS}}_{\mathcal{U}_n^1,V}-T(F)|=\mathcal{O}(n^{-\frac{2s}{2s+d}}+n^{-\frac{1}{2}}).
\end{equation}

For the distributional convergence part, it is enough to show that  $\sqrt{n}(\hat{T}^{\text{DS}}_{\mathcal{U}_n^1,V}-\hat{T}^{\text{DS}}_{\mathcal{U}_n^1})\stackrel{p}{\to}0$ and $$ \sqrt{n}(\hat{T}^{\text{DS}}_{\mathcal{U}_n^1,V}-T(F))\xrightarrow{d} N\left(0,(1-\pi)\mathbb V_f(\psi(X,f))\right), \text{ as } n\to\infty.$$  The desired result would then follow from Slutsky's theorem. 

Since $H_1$ has fixed finite support $S$, it follows from \cref{lem:rn-to-s-finite} that with probability 1, we eventually have $\hat{T}^{\text{DS}}_{\mathcal{U}_n^1,V}=\hat{T}^{\text{DS}}_{\mathcal{U}_n^1}$ and so, $\sqrt{n}(\hat{T}^{\text{DS}}_{\mathcal{U}_n^1,V}-\hat{T}^{\text{DS}}_{\mathcal{U}_n^1})\stackrel{p}{\to} 0$.

We begin with the following Taylor expansion around $\hat{f}_{\mathcal{U}_{n}^{1,1},V}$ (\cite{NIPS2015_06138bc5}),
\begin{equation}
\label{eq:vme}
T(f) = T(\hat{f}_{\mathcal{U}_{n}^{1,1},V}) + \int \psi(u; \hat{f}_{\mathcal{U}_{n}^{1,1},V})f(u)du + O(\|\hat{f}_{\mathcal{U}_{n}^{1,1},V} - f\|^2).
\end{equation}

First consider $\hat{T}^{\text{DS},1}_{\mathcal{U}_n^1,V}$. We can write
\begin{align*}
&\sqrt{|\mathcal{U}_{n}^{1,2}| \vee 1}\left( \hat{T}^{\text{DS},1}_{\mathcal{U}_n^1,V} - T(f) \right)\\
&= \sqrt{|\mathcal{U}_{n}^{1,2}| \vee 1} \left( T(\hat{f}_{\mathcal{U}_{n}^{1,1},V})+  \frac{1}{|\mathcal{U}_{n}^{1,2}| \vee 1} \sum_{i:X_i\in\mathcal{U}_{n}^{1,2}} \psi(V_i; \hat{f}_{\mathcal{U}_{n}^{1,1},V}) - T(f) \right) \notag \\
&= \sqrt{\frac{1}{|\mathcal{U}_{n}^{1,2}| \vee 1}} \sum_{i:X_i\in\mathcal{U}_{n}^{1,2}} \left[ \psi(V_i; \hat{f}_{\mathcal{U}_{n}^{1,1},V}) - \psi(V_i; f) -\int \psi(u; \hat{f}_{\mathcal{U}_{n}^{1,1},V})f(u)du  \right]  \\
&\quad\quad+ \sqrt{\frac{1}{|\mathcal{U}_{n}^{1,2}| \vee 1}} \sum_{i:X_i\in\mathcal{U}_{n}^{1,2}} \psi(V_i; f)+ \sqrt{|\mathcal{U}_{n}^{1,2}| \vee 1} \cdot O(\|\hat{f}_{\mathcal{U}_{n}^{1,1},V} - f\|^2). 
\end{align*}

In the second step, we used \eqref{eq:vme}.
Above, the third term is $o_P(1)$ as it follows from \eqref{eq:kde-rate} and the assumption $s>d/2$ that $\|\hat{f}_{\mathcal{U}_{n}^{1,1},V} - f\|_2^2 \in o_P(n^{-1/2})$ and from the fact that $\frac{2|\mathcal{U}_{n}^{1,2}|}{n}\stackrel{p}{\to} 1-\pi_1,$ as $n\to\infty$. The first term can also be shown to be $o_P(1)$ via Chebyshev’s inequality, since
\begin{align}
&\mathbb{V}\left(\sqrt{\frac{2}{n}} \sum_{i:X_i\in\mathcal{U}_{n}^{1,2}} \left[ \psi(V_i; \hat{f}_{\mathcal{U}_{n}^{1,1},V}) - \psi(V_i; f)  -\int \psi(u; \hat{f}_{\mathcal{U}_{n}^{1,1},V})f(u)du \right]\Bigg| V_1^{n/2}\right) \notag \\
&=\frac{2}{n}\mathbb{V}\left(\sum_{i>n/2} \left[ \psi(V_i; \hat{f}_{\mathcal{U}_{n}^{1,1},V}) - \psi(V_i; f) -\int \psi(u; \hat{f}_{\mathcal{U}_{n}^{1,1},V})f(u)du  \right]\mathds{1}(X_i\in\mathcal{U}_{n}^{1,2}) \Bigg| V_1^{n/2}\right) \notag \\
&=\mathbb{V}\left[(\psi(V; \hat{f}_{\mathcal{U}_{n}^{1,1},V}) - \psi(V; f))\mathds{1}(X\in\mathcal{U}_{n}^{1,2})  \Bigg| V_1^{n/2}\right] \notag \\
&\leq \mathbb{E} \left[ \left( \psi(V; \hat{f}_{\mathcal{U}_{n}^{1,1},V}) - \psi(V; f) \right)^2 \Bigg| V_1^{n/2}\right] = O(\|\hat{f}_{\mathcal{U}_{n}^{1,1},V} - f\|^2_2) 
\end{align}
where the last step follows from Assumption 4.1. 
 Hence we have

\begin{equation}
\sqrt{\frac{n}{2}}  \left( \hat{T}^{\text{DS},1}_{\mathcal{U}_n^1,V} - T(f) \right) = {\frac{n/2}{{|\mathcal{U}_{n}^{1,2}| \vee 1}}} \sqrt{\frac{2}{n}}\sum_{i:X_i\in\mathcal{U}_{n}^{1,2}} \psi(V_i; f) + o_P(1)
\end{equation}

We can similarly show

\begin{equation}
\sqrt{\frac{n}{2}}  \left( \hat{T}^{\text{DS},2}_{\mathcal{U}_n^1,V} - T(f) \right) = {\frac{n/2}{{|\mathcal{U}_{n}^{1,1}| \vee 1}}} \sqrt{\frac{2}{n}}\sum_{i:X_i\in\mathcal{U}_{n}^{1,1}} \psi(V_i; f) + o_P(1)
\end{equation}

 We have
\begin{align*}
&\sqrt{n}  \left( \hat{T}^{\text{DS}}_{\mathcal{U}_n^1,V} - T(f) \right)\\
&=\frac{1}{\sqrt{2}}\left[\sqrt{\frac{n}{2}}  \left( \hat{T}^{\text{DS}}_{\mathcal{U}_n^1,V} - T(f) \right)+\sqrt{\frac{n}{2}}  \left( \hat{T}^{\text{DS}}_{\mathcal{U}_n^1,V} - T(f) \right)\right]\\
&=\frac{1}{\sqrt{1-\pi_1}}\times\frac{1}{\sqrt{n(1-\pi_1)}}\sum_{i:X_i\in\mathcal{U}_{n}^{1}} \psi(V_i; f) + \left({\frac{n/2}{{|\mathcal{U}_{n}^{1,1}| \vee 1}}}-\frac{1}{1-\pi_1}\right) \sqrt{\frac{1}{n}}\sum_{i:X_i\in\mathcal{U}_{n}^{1,1}} \psi(V_i; f)\\
&\quad+ \left({\frac{n/2}{{|\mathcal{U}_{n}^{1,2}| \vee 1}}}-\frac{1}{1-\pi_1}\right) \sqrt{\frac{1}{n}}\sum_{i:X_i\in\mathcal{U}_{n}^{1,2}} \psi(X_i; f)+ o_P(1)
\end{align*}
Since, ${\frac{{|\mathcal{U}_{n}^{1,1}| \vee 1}}{n/2}}\stackrel{p}{\to}1-\pi_1$, ${\frac{{|\mathcal{U}_{n}^{1,2}| \vee 1}}{n/2}}\stackrel{p}{\to}1-\pi_1$ and ${\frac{{|\mathcal{U}_{n}^{1}|}}{n}}\stackrel{p}{\to}1-\pi_1$, second and third term above are also $o_P(1)$, and by using random-index central limit theorem and Slutsky’s theorem, we obtain
\begin{equation}
\sqrt{n}  \left( \hat{T}^{\text{DS}}_{\mathcal{U}_n^1,V} - T(f) \right)    \xrightarrow{d}N\left(0,\frac{1}{1-\pi_1}\mathbb V_f(\psi(X,f))\right).
\end{equation}
Therefore, again using Slutsky,
\begin{equation}
\sqrt{n}  \left( \hat{T}^{\text{DS}}_{\mathcal{U}_n^1} - T(f) \right)    \xrightarrow{d}N\left(0,\frac{1}{1-\pi_1}\mathbb V_f(\psi(X,f))\right).
\end{equation}
\end{proof}

\subsubsection{Proof of \Cref{thm:conv-ds-2}}

\begin{proof}
Define,
\begin{equation}
\hat{T}^{\text{DS},1}_{\mathcal{U}_n^1,\mathcal{U}_m^2,V,W}=T(\hat{f}_{\mathcal{U}_n^{1,1},V},\hat{g}_{\mathcal{U}_m^{2,1},W})+\frac{\sum_{i:X_i \in \mathcal{U}_n^{1,2}}\psi_f(V_i;\hat{f}_{\mathcal{U}_n^{1,1},V},\hat{g}_{\mathcal{U}_m^{2,1},W})}{|\mathcal{U}_n^{1,2}| \vee 1}+\frac{\sum_{i:Y_i \in \mathcal{U}_m^{2,2}}\psi_g(W_i;\hat{f}_{\mathcal{U}_n^{1,1},V},\hat{g}_{\mathcal{U}_m^{2,1},W})}{|\mathcal{U}_m^{2,2}| \vee 1}.
\end{equation}
We write 
\begin{align*}
\hat{T}^{\text{DS},1}_{\mathcal{U}_n^1,\mathcal{U}_m^2,V,W}-\hat{T}^{\text{DS},1}_{\mathcal{U}_n^1,\mathcal{U}_m^2}&=T(\hat{f}_{\mathcal{U}_n^{1,1},V},\hat{g}_{\mathcal{U}_m^{2,1},W})-T(\hat{f}_{\mathcal{U}_n^{1,1}},\hat{g}_{\mathcal{U}_m^{2,1}})\\
&+\frac{\sum_{i:X_i \in \mathcal{U}_n^{1,2}\cap\mathcal{S}}\left(\psi_f(V_i;\hat{f}_{\mathcal{U}_n^{1,1},V},\hat{g}_{\mathcal{U}_m^{2,1},W})-\psi_f(X_i;\hat{f}_{\mathcal{U}_n^{1,1}},\hat{g}_{\mathcal{U}_m^{2,1}})\right)}{|\mathcal{U}_n^{1,2}| \vee 1}\\
&+\frac{\sum_{i:X_i \in \mathcal{U}_n^{1,2}\cap\mathcal{S}^c}\left(\psi_f(V_i;\hat{f}_{\mathcal{U}_n^{1,1},V},\hat{g}_{\mathcal{U}_m^{2,1},W})-\psi_f(X_i;\hat{f}_{\mathcal{U}_n^{1,1}},\hat{g}_{\mathcal{U}_m^{2,1}})\right)}{|\mathcal{U}_n^{1,2}| \vee 1}\\
&+\frac{\sum_{i:Y_i \in \mathcal{U}_m^{2,2}\cap\mathcal{S}}\left(\psi_g(W_i;\hat{f}_{\mathcal{U}_n^{1,1},V},\hat{g}_{\mathcal{U}_m^{2,1},W})-\psi_g(Y_i;\hat{f}_{\mathcal{U}_n^{1,1}},\hat{g}_{\mathcal{U}_m^{2,1}})\right)}{|\mathcal{U}_m^{2,2}| \vee 1}\\
&+\frac{\sum_{i:Y_i \in \mathcal{U}_m^{2,2}\cap\mathcal{S}^c}\left(\psi_g(W_i;\hat{f}_{\mathcal{U}_n^{1,1},V},\hat{g}_{\mathcal{U}_m^{2,1},W})-\psi_g(Y_i;\hat{f}_{\mathcal{U}_n^{1,1}},\hat{g}_{\mathcal{U}_m^{2,1}})\right)}{|\mathcal{U}_m^{2,2}| \vee 1}
\end{align*}
Now, each of these five terms is dealt with similarly as we did in the proof of Theorem 4.2 to show
\begin{equation}
\label{eq:1}
 \mathbb E |\hat{T}^{\text{LOO}}_{\mathcal{U}_n^1,\mathcal{U}_m^2,V,W}-\hat{T}^{\text{LOO}}_{\mathcal{U}_n^1,\mathcal{U}_m^2}|\to 0, \text{ as } n\to \infty
\end{equation}
and when $H_1$ has finite support $\mathcal{S}_n$,
\begin{equation}
\label{eq:2}
    \mathbb E |\hat{T}^{\text{DS}}_{\mathcal{U}_n^1,\mathcal{U}_m^2,V,W}-\hat{T}^{\text{DS}}_{\mathcal{U}_n^1,\mathcal{U}_m^2}|=\mathcal{O}(n^{-\frac{2s}{2s+d}}+m^{-\frac{2s}{2s+d}}).
\end{equation}
Since $\{V_i:1\leq i\leq n, X_i\in\mathcal{U}_n^1\}$ are i.i.d. having Lebesgue density $f$ (because $V_i$ and $[X_i\in\mathcal{U}_n^1]$ are independent) and $|\mathcal{U}_n^1|\leq n-\sum_{i=1}^n \Lambda_i$,$|\mathcal{U}_m^2|\leq m-\sum_{i=1}^m \Gamma_i$, under same assumptions, from Theorem 7 of \cite{NIPS2015_06138bc5},
we have that 
\begin{align*}
\mathbb E|\hat{T}^{\text{DS}}_{n,m} - T(F,G)|^2=\mathcal{O}(n^{-\frac{4s}{2s+d}}+n^{-1}+m^{-\frac{4s}{2s+d}}+m^{-1}), \text{ as } n,m \to \infty.
\end{align*}
\begin{align*}
&\mathbb E|\hat{T}^{\text{DS}}_{\mathcal{U}_n^1,\mathcal{U}_m^2,V,W} - T(F,G)|^2\\
&=\sum_{k,l=1}^\infty \mathbb E\left(|\hat{T}^{\text{DS}}_{\mathcal{U}_n^1,\mathcal{U}_m^2,V,W} - T(F,G)|^2\mathds{1}(|\mathcal{U}_n^1|=k,|\mathcal{U}_m^2|=l)\right)\\
&=\sum_{\substack{k\leq n(1-\pi_1)/2,\\\text{or }l\leq m(1-\pi_2)/2}} \mathbb E|\hat{T}^{\text{DS}}_{k,l} - T(F,G)|^2 \times \P(|\mathcal{U}_n^1|=k)\P(|\mathcal{U}_m^2|=l)\\
&\quad+\sum_{\substack{k>n(1-\pi_1)/2,\\\text{and }l>m(1-\pi_2)/2}}\mathbb E|\hat{T}^{\text{DS}}_{k,l} - T(F,G)|^2 \times \P(|\mathcal{U}_n^1|=k)\P(|\mathcal{U}_m^2|=l)\\
&\leq\sup_{k,l}\mathbb E|\hat{T}^{\text{DS}}_{k,l} - T(F,G)|^2 \P(|\mathcal{U}_n^1|\leq n(1-\pi_1)/2)\P(|\mathcal{U}_m^2|\leq m(1-\pi_2)/2)+\mathcal{O}(n^{-\frac{4s}{2s+d}}+n^{-1}+m^{-\frac{4s}{2s+d}}+m^{-1})\\
&\leq\sup_{k,l}\mathbb E|\hat{T}^{\text{DS}}_{k,l} - T(F,G)|^2   \P(n-\sum_{i=1}^n \Lambda_i\leq n(1-\pi_1)/2)\P(m-\sum_{i=1}^m \Gamma_i\leq m(1-\pi_2)/2)\\
&\quad\quad+\mathcal{O}(n^{-\frac{4s}{2s+d}}+n^{-1}+m^{-\frac{4s}{2s+d}}+m^{-1})\\
&\leq\sup_{k,l}\mathbb E|\hat{T}^{\text{DS}}_{k,l} - T(F,G)|^2   \P(\sum_{i=1}^n \Lambda_i-n\pi_1\geq n(1-\pi_1)/2)\P(\sum_{i=1}^m \Gamma_i-m\pi_2\geq m(1-\pi_2)/2)\\
&\quad\quad+\mathcal{O}(n^{-\frac{4s}{2s+d}}+n^{-1}+m^{-\frac{4s}{2s+d}}+m^{-1})\\
&\leq\sup_{k,l}\mathbb E|\hat{T}^{\text{DS}}_{k,l} - T(F,G)|^2  \exp(-n(1-\pi_1)^2/2-m(1-\pi_2)^2/2)+\mathcal{O}(n^{-\frac{4s}{2s+d}}+n^{-1}+m^{-\frac{4s}{2s+d}}+m^{-1}),
\end{align*}
where the last step follows from Hoeffding's bound.
Since $\mathbb E|\hat{T}^{\text{DS}}_{k,l} - T(F,G)|^2 \to 0$ as $n\to \infty$, we have $\sup_{k,l}\mathbb E|\hat{T}^{\text{DS}}_{k,l} - T(F,G)|^2 <\infty$ and $\exp(-n(1-\pi_1)^2/2-m(1-\pi_2)^2/2)\leq \mathcal{O}(n^{-\frac{4s}{2s+d}}+n^{-1}+m^{-\frac{4s}{2s+d}}+m^{-1})$.
Hence,    
\begin{equation}
    \mathbb E|\hat{T}^{\text{DS}}_{{\mathcal{U}_n^1,\mathcal{U}_m^2,V,W}} - T(F,G)|^2=\mathcal{O}(n^{-\frac{4s}{2s+d}}+n^{-1}+m^{-\frac{4s}{2s+d}}+m^{-1}), \text{ as } n,m \to \infty.
\end{equation}
Combining the above with \eqref{eq:1} and \eqref{eq:2}, we finally have
\begin{equation}
 \mathbb E |\hat{T}^{\text{DS}}_{\mathcal{U}_n^1,\mathcal{U}_m^2}-T(F,G)|   \to 0, \text{ as } n\to \infty
\end{equation}
and when $H_1$ has finite support $\mathcal{S}_n$,,
\begin{equation}
    \mathbb E |\hat{T}^{\text{DS}}_{\mathcal{U}_n^1,\mathcal{U}_m^2}-T(F,G)|=\mathcal{O}(n^{-\frac{2s}{2s+d}}+n^{-\frac{1}{2}}+m^{-\frac{2s}{2s+d}}+m^{-\frac{1}{2}}).
\end{equation}

For the distributional convergence part, it is enough to show that  $\sqrt{n}(\hat{T}^{\text{DS}}_{\mathcal{U}_n^1,\mathcal{U}_m^2,V,W}-\hat{T}^{\text{DS}}_{\mathcal{U}_n^1,\mathcal{U}_m^2})\stackrel{p}{\to}0$ and $$ \sqrt{n}(\hat{T}^{\text{DS}}_{\mathcal{U}_n^1,\mathcal{U}_m^2,V,W}-T(F,G))\stackrel{d}{\to}N\left(0, \frac{1}{\zeta(1-\pi_1)}\mathbb V_f(\psi_f(X;f,g)))+\frac{1}{(1-\zeta)(1-\pi_2)}\mathbb V_g(\psi_g(X;f,g))\right).$$  The desired result would then follow from Slutsky's theorem. 

Since $H_1$ has fixed finite support $\mathcal{S}$,, it follows from \cref{lem:rn-to-s-finite} that with probability 1, we eventually have $\hat{T}^{\text{DS}}_{\mathcal{U}_n^1,\mathcal{U}_m^2,V,W}=\hat{T}^{\text{DS}}_{\mathcal{U}_n^1,\mathcal{U}_m^2}$ and so, $\sqrt{n}(\hat{T}^{\text{DS}}_{\mathcal{U}_n^1,\mathcal{U}_m^2,V,W}-\hat{T}^{\text{DS}}_{\mathcal{U}_n^1,\mathcal{U}_m^2})\stackrel{p}{\to} 0$.

We begin with the following Taylor expansion around $\hat{f}_{\mathcal{U}_{n}^{1,1},V}$ and $\hat{g}_{\mathcal{U}_{m}^{2,1},W}$(\cite{NIPS2015_06138bc5}),
\begin{align}
\label{eq:vme-2}
\nonumber T(f,g)& = T(\hat{f}_{\mathcal{U}_{n}^{1,1},V},\hat{g}_{\mathcal{U}_{m}^{2,1},W}) + \int \psi_f(u; \hat{f}_{\mathcal{U}_{n}^{1,1},V},\hat{g}_{\mathcal{U}_{m}^{2,1},W})f(u)du + \int \psi_g(u; \hat{f}_{\mathcal{U}_{n}^{1,1},V},\hat{g}_{\mathcal{U}_{m}^{2,1},W})g(u)du\\
&\quad+ O(\|\hat{f}_{\mathcal{U}_{n}^{1,1},V} - f\|^2+\|\hat{g}_{\mathcal{U}_{m}^{2,1},W} - g\|^2).
\end{align}

First consider $\hat{T}^{\text{DS},1}_{\mathcal{U}_n^1,\mathcal{U}_m^2,V,W}$. We can write
\begin{align*}
&\sqrt{N}\left( \hat{T}^{\text{DS},1}_{\mathcal{U}_n^1,\mathcal{U}_m^2,V,W} - T(f,g) \right)\\
&= \sqrt{N} \Bigg( \frac{\sum_{i:X_i \in \mathcal{U}_n^{1,2}}\psi_f(V_i;\hat{f}_{\mathcal{U}_n^{1,1},V},\hat{g}_{\mathcal{U}_m^{2,1},W})}{|\mathcal{U}_n^{1,2}| \vee 1}+\frac{\sum_{i:Y_i \in \mathcal{U}_m^{2,2}}\psi_g(W_i;\hat{f}_{\mathcal{U}_n^{1,1},V},\hat{g}_{\mathcal{U}_m^{2,1},W})}{|\mathcal{U}_m^{2,2}| \vee 1}\\
&\quad- \int \psi_f(u; \hat{f}_{\mathcal{U}_{n}^{1,1},V},\hat{g}_{\mathcal{U}_{m}^{2,1},W})f(u)du - \int \psi_g(u; \hat{f}_{\mathcal{U}_{n}^{1,1},V},\hat{g}_{\mathcal{U}_{m}^{2,1},W})g(u)du -O(\|\hat{f}_{\mathcal{U}_{n}^{1,1},V} - f\|^2+\|\hat{g}_{\mathcal{U}_{m}^{2,1},W} - g\|^2)\Bigg) \notag \\
&=\frac{\sqrt{N}}{|\mathcal{U}_n^{1,2}| \vee 1}\sum_{i:X_i \in \mathcal{U}_n^{1,2}}\left(\psi_f(V_i;\hat{f}_{\mathcal{U}_n^{1,1},V},\hat{g}_{\mathcal{U}_m^{2,1},W})-\psi_f(V_i;f,g)-\int \psi_f(u; \hat{f}_{\mathcal{U}_{n}^{1,1},V},\hat{g}_{\mathcal{U}_{m}^{2,1},W})f(u)du \right)\\
&+\frac{\sqrt{N}}{|\mathcal{U}_m^{2,2}| \vee 1}\sum_{i:Y_i \in \mathcal{U}_m^{2,2}}\left(\psi_g(W_i;\hat{f}_{\mathcal{U}_n^{1,1},V},\hat{g}_{\mathcal{U}_m^{2,1},W})-\psi_g(W_i;f,g)-\int \psi_f(u; \hat{f}_{\mathcal{U}_{n}^{1,1},V},\hat{g}_{\mathcal{U}_{m}^{2,1},W})g(u)du \right)\\
&+ \frac{\sqrt{N}}{(|\mathcal{U}_n^{1,2}| \vee 1)}\sum_{i:X_i \in \mathcal{U}_n^{1,2}} \psi_f(V_i;f,g)+ \frac{\sqrt{N}}{(|\mathcal{U}_m^{2,2}| \vee 1)}\sum_{i:Y_i \in \mathcal{U}_m^{2,2}} \psi_g(W_i;f,g)\\
&+\sqrt{N}\times O(\|\hat{f}_{\mathcal{U}_{n}^{1,1},V} - f\|^2+\|\hat{g}_{\mathcal{U}_{m}^{2,1},W} - g\|^2).
\end{align*}
Above, the last term is $o_P(1)$ as it follows from \eqref{eq:kde-rate} and the assumption $s>d/2$ that $\|\hat{f}_{\mathcal{U}_{n}^{1,1},V} - f\|_2^2 = o_P(n^{-1/2})$ and $\|\hat{g}_{\mathcal{U}_{m}^{2,1},W} - g\|^2_2=o_P(m^{-1/2})$, as $n,m\to\infty$. The first and second terms can also be shown to be $o_P(1)$ via Chebyshev’s inequality and using assumption 4.1. Therefore, 
\begin{equation}
\sqrt{N}\left( \hat{T}^{\text{DS},1}_{\mathcal{U}_n^1,\mathcal{U}_m^2,V,W} - T(f,g) \right)=\frac{\sqrt{N}}{|\mathcal{U}_n^{1,2}| \vee 1}\sum_{i:X_i \in \mathcal{U}_n^{1,2}} \psi_f(V_i;f,g)+ \frac{\sqrt{N}}{|\mathcal{U}_m^{2,2}| \vee 1}\sum_{i:Y_i \in \mathcal{U}_m^{2,2}} \psi_g(W_i;f,g)+o_P(1)
\end{equation}
Similarly,
\begin{equation}
\sqrt{N}\left( \hat{T}^{\text{DS},2}_{\mathcal{U}_n^1,\mathcal{U}_m^2,V,W} - T(f,g) \right)=\frac{\sqrt{N}}{|\mathcal{U}_n^{1,1}| \vee 1}\sum_{i:X_i \in \mathcal{U}_n^{1,1}} \psi_f(V_i;f,g)+ \frac{\sqrt{N}}{|\mathcal{U}_m^{2,1}| \vee 1}\sum_{i:Y_i \in \mathcal{U}_m^{2,1}} \psi_g(W_i;f,g)+o_P(1)
\end{equation}
Hence,
\begin{align*}
    \sqrt{N}\left( \hat{T}^{\text{DS}}_{\mathcal{U}_n^1,\mathcal{U}_m^2,V,W} - T(f,g) \right)&=\sqrt{\frac{N}{n(1-\pi_1)}}\frac{1}{\sqrt{n(1-\pi_1)}}\sum_{i:X_i \in \mathcal{U}_n^{1}} \psi_f(V_i;f,g)+ \\
    &+\sqrt{\frac{N}{m(1-\pi_2)}}\frac{1}{\sqrt{m(1-\pi_2)}}\sum_{i:Y_i \in \mathcal{U}_m^{2}} \psi_g(W_i;f,g)\\
 &+\left(\frac{N}{2(|\mathcal{U}_n^{1,2}| \vee 1)}-\frac{N}{n(1-\pi_1)}\right)N^{-1/2}\sum_{i:X_i \in \mathcal{U}_n^{1,2}} \psi_f(V_i;f,g)\\
 &+\left(\frac{N}{2(|\mathcal{U}_n^{1,1}| \vee 1)}-\frac{N}{n(1-\pi_1)}\right)N^{-1/2}\sum_{i:X_i \in \mathcal{U}_n^{1,1}} \psi_f(V_i;f,g)\\
  &+ \left(\frac{N}{2(|\mathcal{U}_m^{2,2}| \vee 1)}-\frac{N}{m(1-\pi_2)}\right)N^{-1/2}\sum_{i:Y_i \in \mathcal{U}_m^{2,2}} \psi_g(W_i;f,g) \\
 &+ \left(\frac{N}{2(|\mathcal{U}_m^{2,1}| \vee 1)}-\frac{N}{m(1-\pi_2)}\right)N^{-1/2}\sum_{i:Y_i \in \mathcal{U}_m^{2,1}} \psi_g(W_i;f,g)  +o_P(1).
\end{align*}

Since, for $i=1,2,$ ${\frac{{2(|\mathcal{U}_{n}^{1,i}| \vee 1)}}{n}}\stackrel{p}{\to}1-\pi_1$, ${\frac{2{|\mathcal{U}_{m}^{2,i}| \vee 1}}{m}}\stackrel{p}{\to}1-\pi_2$, ${\frac{{|\mathcal{U}_{n}^{1}|}}{n}}\stackrel{p}{\to}1-\pi_1$ and ${\frac{{|\mathcal{U}_{m}^{2}|}}{m}}\stackrel{p}{\to}1-\pi_2$, all terms except first and second term above are  $o_P(1)$, and by using random-index central limit theorem and Slutsky’s theorem, we obtain
\begin{equation}
 \sqrt{n}(\hat{T}^{\text{DS}}_{\mathcal{U}_n^1,\mathcal{U}_m^2,V,W}-T(F,G))\stackrel{d}{\to}N\left(0, \frac{1}{\zeta(1-\pi_1)}\mathbb V_f(\psi_f(X;f,g)))+\frac{1}{(1-\zeta)(1-\pi_2)}\mathbb V_g(\psi_g(X;f,g))\right).
\end{equation}
Therefore, again using Slutsky,
\begin{equation}
 \sqrt{n}(\hat{T}^{\text{DS}}_{\mathcal{U}_n^1,\mathcal{U}_m^2}-T(F,G))\stackrel{d}{\to}N\left(0, \frac{1}{\zeta(1-\pi_1)}\mathbb V_f(\psi_f(X;f,g)))+\frac{1}{(1-\zeta)(1-\pi_2)}\mathbb V_g(\psi_g(X;f,g))\right).
\end{equation}
\end{proof}

\subsubsection{Proof of \Cref{thm:conv-loo}}

\begin{proof}
Define,
\begin{equation}
\hat{T}^{\text{LOO}}_{\mathcal{U}_n^1,V} = \frac{1}{|\mathcal{U}_n^1| \vee 1} \sum_{i:X_i \in \mathcal{U}_n^1} \left(T(\hat{f}_{\mathcal{U}_n^1,V}^{(-i)})+\psi(V_i;\hat{f}_{\mathcal{U}_n^1,V}^{(-i)})\right)
\end{equation}
Since $\{V_i:1\leq i\leq n, X_i\in\mathcal{U}_n^1\}$ are i.i.d. having Lebesgue density $f$ (because $V_i$ and $[X_i\in\mathcal{U}_n^1]$ are independent) and $|\mathcal{U}_n^1|\leq n-\sum_{i=1}^n \Lambda_i$, as $n\to\infty$, under same assumptions, from Theorem 5 of \cite{NIPS2015_06138bc5},
we have that 
\begin{align*}
    \mathbb E|\hat{T}^{\text{LOO}}_{n} - T(F)|^2=\mathcal{O}(n^{-\frac{4s}{2s+d}}+n^{-1}), \text{ as } n \to \infty.
\end{align*}
\begin{align*}
    &\mathbb E|\hat{T}^{\text{LOO}}_{\mathcal{U}_n^1,V} - T(F)|^2\\
    &=\sum_{k=1}^\infty \mathbb E\left(|\hat{T}^{\text{LOO}}_{\mathcal{U}_n^1,V} - T(F)|^2\mathds{1}(|\mathcal{U}_n^1|=k)\right)\\
    &=\sum_{k\leq n(1-\pi)/2} \mathbb E|\hat{T}^{\text{LOO}}_{k} - T(F)|^2 \times \P(|\mathcal{U}_n^1|=k)+\sum_{k>n(1-\pi)/2} \mathbb E|\hat{T}^{\text{LOO}}_{k} - T(F)|^2 \times \P(|\mathcal{U}_n^1|=k)\\
    &\leq\sup_{k}\mathbb E|\hat{T}^{\text{LOO}}_{k} - T(F)|^2 \sum_{k\leq n(1-\pi)/2}  \P(|\mathcal{U}_n^1|=k)+\mathcal{O}(n^{-\frac{4s}{2s+d}}+n^{-1})\times\sum_{k>n(1-\pi)/2} \P(|\mathcal{U}_n^1|=k)\\
    &\leq\sup_{k}\mathbb E|\hat{T}^{\text{LOO}}_{k} - T(F)|^2   \P(n-\sum_{i=1}^n \Lambda_i\leq n(1-\pi)/2)+\mathcal{O}(n^{-\frac{4s}{2s+d}}+n^{-1})\\
    &\leq\sup_{k}\mathbb E|\hat{T}^{\text{LOO}}_{k} - T(F)|^2   \P(\sum_{i=1}^n \Lambda_i-n\pi\geq n(1-\pi)/2)+\mathcal{O}(n^{-\frac{4s}{2s+d}}+n^{-1})\\
    &\leq\sup_{k}\mathbb E|\hat{T}^{\text{LOO}}_{k} - T(F)|^2  \exp(-n(1-\pi)^2/2)+\mathcal{O}(n^{-\frac{4s}{2s+d}}+n^{-1}) ~~~[\text{By Hoeffding bound}]
\end{align*}
Since $\mathbb E|\hat{T}^{\text{LOO}}_{k} - T(F)|^2 \to 0$ as $n\to \infty$, we have $\sup_{k}\mathbb E|\hat{T}^{\text{LOO}}_{k} - T(F)|^2 <\infty$ and $\exp(-n(1-\pi)^2/2)\leq \mathcal{O}(n^{-\frac{4s}{2s+d}}+n^{-1})$.
Hence,
\begin{equation}
\label{eq:original-order}
  \mathbb E|\hat{T}^{\text{LOO}}_{\mathcal{U}_n^1,V} - T(F)|^2  =\mathcal{O}(n^{-\frac{4s}{2s+d}}+n^{-1})
\end{equation}
  Therefore, to show $\mathbb E|\hat{T}^{\text{LOO}}_{\mathcal{U}_n^1} - T(F)|\to 0$, it is enough to show that, $\mathbb E|\hat{T}^{\text{LOO}}_{\mathcal{U}_n^1,V}-\hat{T}^{\text{LOO}}_{\mathcal{U}_n^1}|   \to 0$, as $n\to\infty.$
  Now,
\begin{align}
\label{eq:break-loo1}
\nonumber|\hat{T}^{\text{LOO}}_{\mathcal{U}_n^1,V}-\hat{T}^{\text{LOO}}_{\mathcal{U}_n^1}|
&\leq\frac{1}{|\mathcal{U}_n^1| \vee 1}\Bigg[\bigg|\sum_{\substack{1\leq i\leq n:\\X_i \in \mathcal{U}_n^1}} \left(T(\hat{f}_{\mathcal{U}_n^1,V}^{(-i)})-T(\hat{f}_{\mathcal{U}_n^1}^{(-i)})\right)\bigg|+\bigg|\sum_{\substack{1\leq i\leq n:\\X_i \in \mathcal{U}_n^1\cap\mathcal{S}}}\left( \psi(V_i;\hat{f}_{\mathcal{U}_n^1,V}^{(-i)})-\psi(X_i;\hat{f}_{\mathcal{U}_n^1}^{(-i)})\right)\bigg|\\
 &+\bigg|\sum_{\substack{1\leq i\leq n:\\X_i \in \mathcal{U}_n^1\cap\mathcal{S}^c}}\left( \psi(V_i;\hat{f}_{\mathcal{U}_n^1,V}^{(-i)})-\psi(X_i;\hat{f}_{\mathcal{U}_n^1}^{(-i)})\right)\bigg|\Bigg]
\end{align}
From \eqref{eq:f-hat-diff}, it follows that
\begin{equation}
\label{idiff}
\int|\hat{f}_{\mathcal{U}_n^1,V}^{(-i)}(x)-\hat{f}_{\mathcal{U}_n^1}^{(-i)}(x)|dx\stackrel{a.s.}{\leq}\frac{2}{(|\mathcal{U}_n^1|-1)\vee 1}\sum_{{j=1,\neq i}}^n\mathds{1}(X_j \in \mathcal{S}\setminus\mathcal{R}_n)
\end{equation}
For the first term:
  \begin{align*}
\frac{1}{|\mathcal{U}_n^1| \vee 1}\sum_{\substack{1\leq i\leq n:\\X_i \in \mathcal{U}_n^1}} \left(T(\hat{f}_{\mathcal{U}_n^1,V}^{(-i)})-T(\hat{f}_{\mathcal{U}_n^1}^{(-i)})\right)&\leq\frac{1}{|\mathcal{U}_n^1| \vee 1}\sum_{\substack{1\leq i\leq n:\\X_i \in \mathcal{U}_n^1}}  L_\phi L_\nu\int \left|\hat{f}_{\mathcal{U}_n^1,V}^{(-i)}(x)-\hat{f}_{\mathcal{U}_n^1}^{(-i)}(x)\right|dx\\
&\leq\frac{L_\phi L_\nu }{|\mathcal{U}_n^1| \vee 1}\sum_{\substack{1\leq i\leq n:\\X_i \in \mathcal{U}_n^1}}\frac{2}{(|\mathcal{U}_n^1|-1)\vee 1 }\sum_{j=1,\neq i}^n\mathds{1}(X_j \in \mathcal{U}_n^1\cap\mathcal{S})\\
&\leq 2L_\phi L_\nu \times\frac{1}{(|\mathcal{U}_n^1|-1)\vee 1}\times\sum_{j=1}^n\mathds{1}(X_j \in \mathcal{S}\setminus\mathcal{R}_n),
\end{align*}
where $L_\phi$ and $L_\nu$ are the Lipschitz constants for the functions $\phi$ and $\nu$ respectively and the last inequality follows from the fact that $\mathcal{U}_n^1\cap\mathcal{S}\subseteq\mathcal{S}\setminus\mathcal{R}_n$. Hence, from the above calculation, we have
\begin{equation}
\label{eq:term1-loo1}
    \mathbb E\left|\frac{1}{|\mathcal{U}_n^1| \vee 1}\sum_{\substack{1\leq i\leq n:\\X_i \in \mathcal{U}_n^1}} \left(T(\hat{f}_{\mathcal{U}_n^1,V}^{(-i)})-T(\hat{f}_{\mathcal{U}_n^1}^{(-i)})\right)\right|\leq  2L_\phi L_\nu\times n\mathbb E\left[\frac{\mathds{1}(X_1 \in \mathcal{S}\setminus\mathcal{R}_n)}{(|\mathcal{U}_n^1|-1) \vee 1}\right].
\end{equation}
For the expectation of the second term,
  \begin{align}
  \label{eq:term2-loo1}
\nonumber\mathbb E\left(\frac{1}{|\mathcal{U}_n^1| \vee 1}\sum_{\substack{1\leq i\leq n:\\X_i \in \mathcal{U}_n^1\cap\mathcal{S}}}\left|\psi(V_i;\hat{f}_{\mathcal{U}_n^1,V}^{(-i)})-\psi(X_i;\hat{f}_{\mathcal{U}_n^1}^{(-i)})\right|\right)&\leq \mathbb E\left(\frac{2\|\psi\|_\infty}{|\mathcal{U}_n^1| \vee 1}\sum_{i=1}^n\mathds{1}(X_i \in \mathcal{U}_n^1\cap\mathcal{S})\right)\\
&\leq 2\|\psi\|_\infty\times n\mathbb E\left[\frac{\mathds{1}(X_1 \in \mathcal{S}\setminus\mathcal{R}_n)}{(|\mathcal{U}_n^1|-1) \vee 1}\right].
\end{align}

Finally, we show that the expectation of the last term,
\begin{align}
\label{term3-loo1}
\nonumber\mathbb E\left|\frac{1}{|\mathcal{U}_n^1| \vee 1}\sum_{\substack{1\leq i\leq n:\\X_i \in \mathcal{U}_n^1\cap\mathcal{S}^c}}\left( \psi(V_i;\hat{f}_{\mathcal{U}_n^1,V}^{(-i)})-\psi(X_i;\hat{f}_{\mathcal{U}_n^1}^{(-i)})\right)\right|
  &{\leq}\sum_{\substack{1\leq i\leq n}}\mathbb E\left|\frac{\psi(V_i;\hat{f}_{\mathcal{U}_n^1,V}^{(-i)})-\psi(V_i;\hat{f}_{\mathcal{U}_n^1}^{(-i)})}{|\mathcal{U}_n^1| \vee 1} \right|\\
&=n \mathbb E\left|\frac{\psi(V_1;\hat{f}_{\mathcal{U}_n^1,V}^{(-1)})-\psi(V_1;\hat{f}_{\mathcal{U}_n^1}^{(-1)})}{|\mathcal{U}_n^1| \vee 1} \right|. 
\end{align}
Using Assumption 4.1, for large enough $n$, $\mathbb E\left(\left|\psi(V_1;\hat{f}_{\mathcal{U}_n^1,V}^{(-1)})-\psi(V_1;\hat{f}_{\mathcal{U}_n^1}^{(-1)}) \right|\Bigg\vert X_{-1},V_{-1}\right)\leq C|\hat{f}_{\mathcal{U}_n^1,V}^{(-1)}-\hat{f}_{\mathcal{U}_n^1}^{(-1)}|$, for some constant $C$. Therefore,
\begin{align*}
& n\mathbb E\left|\frac{\psi(V_1;\hat{f}_{\mathcal{U}_n^1,V}^{(-1)})-\psi(V_1;\hat{f}_{\mathcal{U}_n^1}^{(-1)})}{|\mathcal{U}_n^1| \vee 1} \right|\\
 &\leq n\mathbb E\left( \mathbb E\left(\left|\frac{\psi(V_1;\hat{f}_{\mathcal{U}_n^1,V}^{(-1)})-\psi(V_1;\hat{f}_{\mathcal{U}_n^1}^{(-1)})}{|\mathcal{U}_n^1| \vee 1} \right|\Bigg\vert X_{-1},V_{-1}\right)\right) \\
 &\leq \mathbb E\left(\frac{n}{(|\mathcal{U}_{n,-1}^1|-1) \vee 1} \mathbb E\left(\left|\psi(V_1;\hat{f}_{\mathcal{U}_n^1,V}^{(-1)})-\psi(V_1;\hat{f}_{\mathcal{U}_n^1}^{(-1)}) \right|\Bigg\vert X_{-1},V_{-1}\right)\right) \\
 &\leq \mathbb E\left( \frac{Cn}{(|\mathcal{U}_{n,-1}^1|-1) \vee 1}\int |\hat{f}_{\mathcal{U}_n^1,V}^{(-1)}-\hat{f}_{\mathcal{U}_n^1}^{(-1)}|\right)\\
&\leq \mathbb E\left[\frac{C}{(|\mathcal{U}_{n,-1}^1|-1) \vee 1}\times\frac{2}{(|\mathcal{U}_n^1|-1) \vee 1}\sum_{{j=1,\neq i}}^n\mathds{1}(X_j \in \mathcal{S}\setminus\mathcal{R}_n)\right]\\
&=2Cn(n-1)\mathbb E\left[\frac{1}{(|\mathcal{U}_{n,-1}^1|-1) \vee 1}\frac{\mathds{1}(X_1 \in \mathcal{S}\setminus\mathcal{R}_n)}{(|\mathcal{U}_n^1|-1) \vee 1}\right],
\end{align*}
where the last inequality follows from \Cref{idiff} and $|\mathcal{U}_{n,-1}^1|$ denotes the number of unique elements in $X_{-1}=\{X_2,\cdots,X_n\}$, and the second inequality follows from the observation that conditioned on $X_{-1}$, $|\mathcal{U}_{n}|\geq |\mathcal{U}_{n,-1}^1|-1$. Now, Cauchy-Schwarz inequality implies
\begin{align*}
  n\mathbb E\left[\frac{\mathds{1}(X_1 \in \mathcal{S}\setminus\mathcal{R}_n)}{(|\mathcal{U}_n^1|-1) \vee 1}\right]\leq\sqrt{ \mathbb E\left[\frac{n^2}{((|\mathcal{U}_n^1|-1) \vee 1)^2}\right]\mathbb E[\mathds{1}(X_1 \in \mathcal{S}\setminus\mathcal{R}_n)]}
\end{align*}
and
\begin{align*}
  \mathbb E\left[\frac{n(n-1)}{(|\mathcal{U}_{n,-1}^1|-1) \vee 1}\frac{\mathds{1}(X_1 \in \mathcal{S}\setminus\mathcal{R}_n)}{(|\mathcal{U}_n^1|-1) \vee 1}\right]\leq\sqrt[3]{\mathbb E\left[\frac{(n-1)^3}{((|\mathcal{U}_{n,-1}^1|-1) \vee 1)^3}\right] \mathbb E\left[\frac{n^3}{((|\mathcal{U}_n^1|-1) \vee 1)^3}\right]\mathbb E[\mathds{1}(X_1 \in \mathcal{S}\setminus\mathcal{R}_n)]}
\end{align*}
 Define, $A_n=\{k\in0,1,\cdots,n-2: |k-n\pi|<n^{2/3}\}$. Also note that $n-|\mathcal{U}_n^1|\leq\sum_{i=1}^n \Lambda_i$ almost surely. 
Therefore, for $\gamma=2,3$,
\begin{align*}
\mathbb E\left[\frac{n^\gamma}{((|\mathcal{U}_n^1|-1) \vee 1)^\gamma}\right]&\leq \mathbb E\left[\frac{n^\gamma}{((n-1-\sum_{i=1}^n \Lambda_i) \vee 1)^\gamma}\right]\\
&\leq \sum_{k=1}^{n-2}\frac{n^\gamma}{(n-1-k)^\gamma}\P\left[\sum_{i=1}^n \Lambda_i=k\right]\\
&\leq\sum_{k\in A_n}\frac{n^\gamma}{(n-1-n\pi-n^{2/3})^\gamma}\P\left[\sum_{i=1}^n \Lambda_i=k\right]+\sum_{k\in A_n^c}n^\gamma\P\left[\sum_{i=1}^n \Lambda_i=k\right]\\
&\leq\frac{n^\gamma}{(n-1-n\pi-n^{2/3})^\gamma}+n^\gamma\P\left[\left|\sum_{i=1}^n \Lambda_i-n\pi\right|\geq n^{2/3}\right]\\
&\leq\frac{n^\gamma}{(n-1-n\pi-n^{2/3})^\gamma}+n^\gamma\exp\{-2n^{1/3}\}\to 1/(1-\pi)^\gamma, \text{ as } n\to \infty.
\end{align*}
Similarly, one can show that $\lim\sup_{n\to\infty} \mathbb E\left[\frac{(n-1)^3}{((|\mathcal{U}_{n,-1}^1|-1) \vee 1)^3}\right]\leq 1/(1-\pi)^3$.

The above calculations, along with \Cref{lem:exp-decay} show that $ \mathbb E\left[\frac{n(n-1)}{(|\mathcal{U}_{n,-1}^1|-1) \vee 1}\frac{\mathds{1}(X_1 \in \mathcal{S}\setminus\mathcal{R}_n)}{(|\mathcal{U}_n^1|-1) \vee 1}\right]\to 0$ and $ n\mathbb E\left[\frac{\mathds{1}(X_1 \in \mathcal{S}\setminus\mathcal{R}_n)}{(|\mathcal{U}_n^1|-1) \vee 1}\right]\to 0$, as $n\to\infty$ and for finite $\mathcal{S}_n$ satisfying the given condition, $ \mathbb E\left[\frac{n(n-1)}{(|\mathcal{U}_{n,-1}^1|-1) \vee 1}\frac{\mathds{1}(X_1 \in \mathcal{S}\setminus\mathcal{R}_n)}{(|\mathcal{U}_n^1|-1) \vee 1}\right]=\mathcal{O}(n^{\frac{-2s}{2s+d}})$ and $ n\mathbb E\left[\frac{\mathds{1}(X_1 \in \mathcal{S}\setminus\mathcal{R}_n)}{(|\mathcal{U}_n^1|-1) \vee 1}\right]=\mathcal{O}(n^{\frac{-3s}{2s+d}})$.  Therefore, from \eqref{eq:term1-loo1}, \eqref{eq:term2-loo1} and \eqref{term3-loo1}, and taking expectation in both sides of \eqref{eq:break-loo1}, we obtain
\begin{equation}
 \mathbb E |\hat{T}^{\text{LOO}}_{\mathcal{U}_n^1,V}-\hat{T}^{\text{LOO}}_{\mathcal{U}_n^1}|   \to 0, \text{ as } n\to \infty
\end{equation}
and when $H_1$ has finite support $\mathcal{S}_n$,
\begin{equation}
    \mathbb E |\hat{T}^{\text{LOO}}_{\mathcal{U}_n^1,V}-\hat{T}^{\text{LOO}}_{\mathcal{U}_n^1}|=\mathcal{O}(n^{-\frac{2s}{2s+d}}).
\end{equation}
Combining the above with \eqref{eq:original-order}, we have
\begin{equation}
 \mathbb E |\hat{T}^{\text{LOO}}_{\mathcal{U}_n^1}-T(F)|   \to 0, \text{ as } n\to \infty
\end{equation}
and when $H_1$ has finite support $\mathcal{S}_n$,
\begin{equation}
    \mathbb E |\hat{T}^{\text{LOO}}_{\mathcal{U}_n^1}-T(F)|=\mathcal{O}(n^{-\frac{2s}{2s+d}}+n^{-\frac{1}{2}}).
\end{equation}

\end{proof}

\subsubsection{Proof of \Cref{thm:conv-loo-2}}

\begin{proof}
Define,
\begin{equation}
\hat{T}^{\text{LOO}}_{\mathcal{U}_n^1,\mathcal{U}_m^2,V,W} =\frac{1}{|\mathcal{U}_n^1|\vee|\mathcal{U}_m^2| \vee 1} \sum_{i=1}^{|\mathcal{U}_n^1|\vee|\mathcal{U}_m^2|}\left(T(\hat{f}_{\mathcal{U}_n^1,V}^{(-j_i)},\hat{g}_{\mathcal{U}_m^2,W}^{(-k_i)})+\psi_f(V_i;\hat{f}_{\mathcal{U}_n^1,V}^{(-j_i)},\hat{g}_{\mathcal{U}_m^2,W}^{(-k_i)})+\psi_g(Z_i;\hat{f}_{\mathcal{U}_n^1,V}^{(-j_i)},\hat{g}_{\mathcal{U}_m^2,Z}^{(-k_i)})\right),
\end{equation}
Since $\{V_i:1\leq i\leq n, X_i\in\mathcal{U}_n^1\}$ are i.i.d. having Lebesgue density $f$ (because $V_i$ and $[X_i\in\mathcal{U}_n^1]$ are independent) and $|\mathcal{U}_n^1|\leq n-\sum_{i=1}^n \Lambda_i$,$|\mathcal{U}_m^2|\leq m-\sum_{i=1}^m \Gamma_i$, under same assumptions, from Theorem 7 of \cite{NIPS2015_06138bc5},
we have that 
\begin{align*}
    \mathbb E|\hat{T}^{\text{LOO}}_{n,m} - T(F,G)|^2=\mathcal{O}(n^{-\frac{4s}{2s+d}}+n^{-1}+m^{-\frac{4s}{2s+d}}+m^{-1}), \text{ as } n,m \to \infty.
\end{align*}
\begin{align*}
    &\mathbb E|\hat{T}^{\text{LOO}}_{\mathcal{U}_n^1,\mathcal{U}_m^2,V,W} - T(F,G)|^2\\
    &=\sum_{k,l=1}^\infty \mathbb E\left(|\hat{T}^{\text{LOO}}_{\mathcal{U}_n^1,\mathcal{U}_m^2,V,W} - T(F,G)|^2\mathds{1}(|\mathcal{U}_n^1|=k,|\mathcal{U}_m^2|=l)\right)\\
    &=\sum_{\substack{k\leq n(1-\pi_1)/2,\\\text{or }l\leq m(1-\pi_2)/2}} \mathbb E|\hat{T}^{\text{LOO}}_{k,l} - T(F,G)|^2 \times \P(|\mathcal{U}_n^1|=k)\P(|\mathcal{U}_m^2|=l)\\
    &\quad+\sum_{\substack{k>n(1-\pi_1)/2,\\\text{and }l>m(1-\pi_2)/2}}\mathbb E|\hat{T}^{\text{LOO}}_{k,l} - T(F,G)|^2 \times \P(|\mathcal{U}_n^1|=k)\P(|\mathcal{U}_m^2|=l)\\
    &\leq\sup_{k,l}\mathbb E|\hat{T}^{\text{LOO}}_{k,l} - T(F,G)|^2 \P(|\mathcal{U}_n^1|\leq n(1-\pi_1)/2)\P(|\mathcal{U}_m^2|\leq m(1-\pi_2)/2)+\mathcal{O}(n^{-\frac{4s}{2s+d}}+n^{-1}+m^{-\frac{4s}{2s+d}}+m^{-1})\\
    &\leq\sup_{k,l}\mathbb E|\hat{T}^{\text{LOO}}_{k,l} - T(F,G)|^2   \P(n-\sum_{i=1}^n \Lambda_i\leq n(1-\pi_1)/2)\P(m-\sum_{i=1}^m \Gamma_i\leq m(1-\pi_2)/2)\\
    &\quad\quad+\mathcal{O}(n^{-\frac{4s}{2s+d}}+n^{-1}+m^{-\frac{4s}{2s+d}}+m^{-1})\\
    &\leq\sup_{k,l}\mathbb E|\hat{T}^{\text{LOO}}_{k,l} - T(F,G)|^2   \P(\sum_{i=1}^n \Lambda_i-n\pi_1\geq n(1-\pi_1)/2)\P(\sum_{i=1}^m \Gamma_i-m\pi_2\geq m(1-\pi_2)/2)\\
    &\quad\quad+\mathcal{O}(n^{-\frac{4s}{2s+d}}+n^{-1}+m^{-\frac{4s}{2s+d}}+m^{-1})\\
    &\leq\sup_{k,l}\mathbb E|\hat{T}^{\text{LOO}}_{k,l} - T(F,G)|^2  \exp(-n(1-\pi_1)^2/2-m(1-\pi_2)^2/2)+\mathcal{O}(n^{-\frac{4s}{2s+d}}+n^{-1}+m^{-\frac{4s}{2s+d}}+m^{-1}),
\end{align*}
where the last step follows from Hoeffding's bound.
Since $\mathbb E|\hat{T}^{\text{LOO}}_{k,l} - T(F,G)|^2 \to 0$ as $n\to \infty$, we have $\sup_{k,l}\mathbb E|\hat{T}^{\text{LOO}}_{k,l} - T(F,G)|^2 <\infty$ and $\exp(-n(1-\pi_1)^2/2-m(1-\pi_2)^2/2)\leq \mathcal{O}(n^{-\frac{4s}{2s+d}}+n^{-1}+m^{-\frac{4s}{2s+d}}+m^{-1})$.
Hence,    
\begin{equation}
\label{eq:original-order-2}
    \mathbb E|\hat{T}^{\text{LOO}}_{{\mathcal{U}_n^1,\mathcal{U}_m^2,V,W}} - T(F,G)|^2=\mathcal{O}(n^{-\frac{4s}{2s+d}}+n^{-1}+m^{-\frac{4s}{2s+d}}+m^{-1}), \text{ as } n,m \to \infty.
\end{equation}
\begin{align}
\label{eq:break-loo2}
\nonumber|\hat{T}^{\text{LOO}}_{\mathcal{U}_n^1,\mathcal{U}_m^2,V,W}-\hat{T}^{\text{LOO}}_{\mathcal{U}_n^1,\mathcal{U}_m^2}|
&\leq\frac{1}{|\mathcal{U}_n^1| \vee |\mathcal{U}_m^2| \vee1}\Bigg[\bigg|\sum_{i=1}^{|\mathcal{U}_n^1| \vee |\mathcal{U}_m^2|} \left(T(\hat{f}_{\mathcal{U}_n^1,V}^{(-j_i)},\hat{g}_{\mathcal{U}_m^2,W}^{(-k_i)})-T(\hat{f}_{\mathcal{U}_n^1}^{(-j_i)},\hat{g}_{\mathcal{U}_m^2}^{(-k_i)})\right)\bigg|\\
\nonumber&+\bigg|\sum_{\substack{1\leq i\leq |\mathcal{U}_n^1| \vee |\mathcal{U}_m^2|:\\X_{j_i} \in \mathcal{U}_n^1\cap\mathcal{S}}}\left( \psi_f(V_i;\hat{f}_{\mathcal{U}_n^1,V}^{(-j_i)},\hat{g}_{\mathcal{U}_m^2,W}^{(-k_i)})-\psi_f(X_i;\hat{f}_{\mathcal{U}_n^1}^{(-j_i)},\hat{g}_{\mathcal{U}_m^2}^{(-k_i)})\right)\bigg|\\
\nonumber&+\bigg|\sum_{\substack{1\leq i\leq |\mathcal{U}_n^1| \vee |\mathcal{U}_m^2|:\\X_{j_i} \in \mathcal{U}_n^1\cap\mathcal{S}^c}}\left( \psi_f(V_i;\hat{f}_{\mathcal{U}_n^1,V}^{(-j_i)},\hat{g}_{\mathcal{U}_m^2,W}^{(-k_i)})-\psi_f(X_i;\hat{f}_{\mathcal{U}_n^1}^{(-j_i)},\hat{g}_{\mathcal{U}_m^2}^{(-k_i)})\right)\bigg|\\
\nonumber &+\bigg|\sum_{\substack{1\leq i\leq |\mathcal{U}_n^1| \vee |\mathcal{U}_m^2|:\\Y_{k_i} \in \mathcal{U}_n^1\cap\mathcal{S}}}\left( \psi_g(Z_i;\hat{f}_{\mathcal{U}_n^1,V}^{(-j_i)},\hat{g}_{\mathcal{U}_m^2,Z}^{(-k_i)})-\psi_g(Y_i;\hat{f}_{\mathcal{U}_n^1}^{(-j_i)},\hat{g}_{\mathcal{U}_m^2}^{(-k_i)})\right)\bigg|\Bigg]\\
&+\bigg|\sum_{\substack{1\leq i\leq |\mathcal{U}_n^1| \vee |\mathcal{U}_m^2|:\\Y_{k_i} \in \mathcal{U}_n^1\cap\mathcal{S}^c}}\left( \psi_g(Z_i;\hat{f}_{\mathcal{U}_n^1,V}^{(-j_i)},\hat{g}_{\mathcal{U}_m^2,Z}^{(-k_i)})-\psi_g(Y_i;\hat{f}_{\mathcal{U}_n^1}^{(-j_i)},\hat{g}_{\mathcal{U}_m^2}^{(-k_i)})\right)\bigg|\Bigg]
\end{align}
Now, each of these five terms is dealt with similarly as we did in the proof of Theorem 4.4 to show
\begin{equation}
 \mathbb E |\hat{T}^{\text{LOO}}_{\mathcal{U}_n^1,\mathcal{U}_m^2,V,W}-\hat{T}^{\text{LOO}}_{\mathcal{U}_n^1,\mathcal{U}_m^2}|\to 0, \text{ as } n\to \infty
\end{equation}
and when $H_1$ has finite support $\mathcal{S}_n$ satisfying the given condition,
\begin{equation}
    \mathbb E |\hat{T}^{\text{LOO}}_{\mathcal{U}_n^1,\mathcal{U}_m^2,V,W}-\hat{T}^{\text{LOO}}_{\mathcal{U}_n^1,\mathcal{U}_m^2}|=\mathcal{O}(n^{-\frac{2s}{2s+d}}+m^{-\frac{2s}{2s+d}}).
\end{equation}
Combining the above with \eqref{eq:original-order-2}, we finally have
\begin{equation}
 \mathbb E |\hat{T}^{\text{LOO}}_{\mathcal{U}_n^1,\mathcal{U}_m^2}-T(F,G)|   \to 0, \text{ as } n\to \infty
\end{equation}
and when $H_1$ has finite support $\mathcal{S}_n$ satisfying the given condition,
\begin{equation}
    \mathbb E |\hat{T}^{\text{LOO}}_{\mathcal{U}_n^1,\mathcal{U}_m^2}-T(F,G)|=\mathcal{O}(n^{-\frac{2s}{2s+d}}+n^{-\frac{1}{2}}+m^{-\frac{2s}{2s+d}}+m^{-\frac{1}{2}}).
\end{equation}
\end{proof}

\section{Additional Experimental Result on Estimation of Rényi Divergence}

We generate two samples of the same size ($n=m$) from two different distributions, each formed as a mixture of a continuous and a discrete component. The discrete part of both distributions is given by a scaled Poisson distribution, $\mathrm{Poisson}(1)/5$, supported on a countable set. The continuous component of the first distribution is the uniform distribution on $[0,1]$, while the continuous component of the second distribution has density $0.5 + 5t^5$ for $t \in [0,1]$.

Our goal is to estimate the Rényi-0.75 divergence between these two mixed distributions using our leave-one-out (LOO) estimator, denoted $\hat{T}^{\text{LOO}}_{\mathcal{U}_n^1,\mathcal{U}_m^2}$. To evaluate its performance, we report the average absolute error across 100 independent runs. We compare our method against two baselines: (a) the LOO estimator from \cite{NIPS2015_06138bc5}, which uses the full data and hence, inconsistent when atoms are present. (a) the oracle estimator: the estimator is the same, but it now has access to the labels indicating whether each point was generated from the continuous component, and uses only the continuous part for estimation. The results, shown in \Cref{fig:renyi}, highlight that the mean absolute error of our method is very close to that of the oracle.

\begin{figure}
    \centering
    \includegraphics[width=0.5\linewidth]{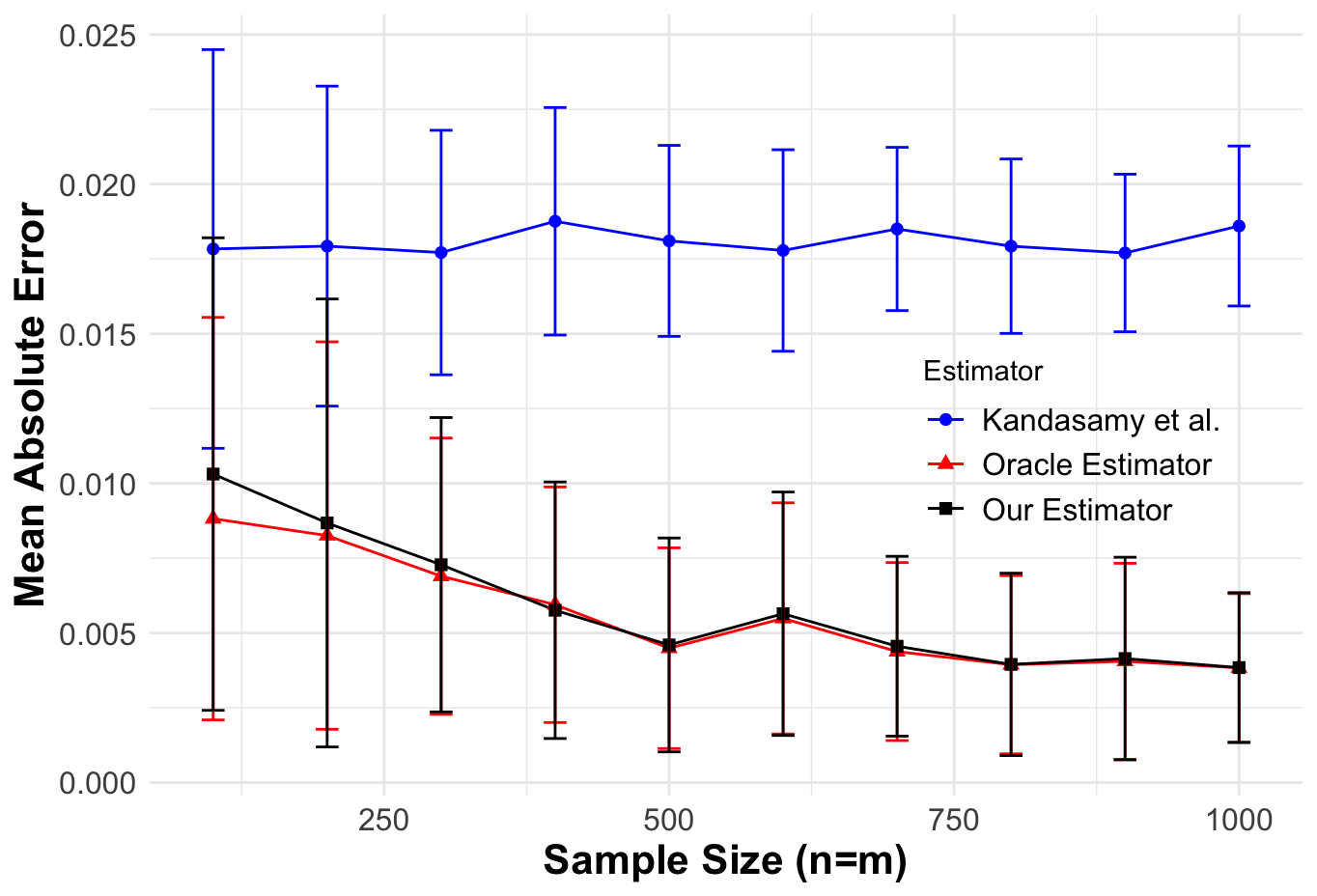}
\caption[]{The average of the absolute error of estimation of Rényi-0.75 divergence is plotted against the sample size, $n=m$. For the first sample, $60\%$ of the data is drawn from the Uniform$(0,1)$ and the remaining $40\%$ from $\mathrm{Poisson}(1)/5$. For the second sample, $60\%$ of the data is drawn from the density $f(t)=0.5+5t^9, t\in [0,1]$ and the remaining $40\%$ from $\mathrm{Poisson}(1)/5$. Our atom-aware estimator closely matches the performance of the oracle that has access to the labels, and their mean absolute error approaches zero as the sample size increases. However, the original estimator of \cite{NIPS2015_06138bc5} fails due to its inability to handle atoms in the distribution.
} 
\label{fig:renyi}  
\end{figure}

\end{document}